\documentclass[10pt]{article}%
\usepackage{amsthm}
\usepackage{amsmath}
\usepackage{amscd}
\usepackage{amsmath}
\usepackage{amsfonts}
\usepackage{amssymb}
\usepackage{graphicx}
\usepackage[margin=0.5in]{geometry}%
\usepackage{amsfonts}
\usepackage{amssymb}
\usepackage{amsthm}
\usepackage{amsmath} 
\usepackage{graphicx}
\usepackage{hyperref}
\usepackage{color}
\usepackage[mathscr]{euscript}

\providecommand{\U}[1]{\protect\rule{.1in}{.1in}}
\providecommand{\U}[1]{\protect\rule{.1in}{.1in}}
\providecommand{\U}[1]{\protect\rule{.1in}{.1in}}
\providecommand{\U}[1]{\protect\rule{.1in}{.1in}}
\providecommand{\U}[1]{\protect\rule{.1in}{.1in}}
\providecommand{\U}[1]{\protect\rule{.1in}{.1in}}
\providecommand{\U}[1]{\protect\rule{.1in}{.1in}}
\providecommand{\U}[1]{\protect\rule{.1in}{.1in}}

\providecommand{\U}[1]{\protect\rule{.1in}{.1in}}

\numberwithin{equation}{section}
\newtheorem{theorem}{Theorem}[section]
\newtheorem{corollary}[theorem]{Corollary}

\newtheorem{proposition}[theorem]{Proposition}

\newtheorem{remark}[theorem]{Remark}

\def\p2{\mathcal A_{\Phi,2\pi}(B)}
\def\0p2{\mathcal A_{\Phi,2\pi}(0)}
\def\sp2{\mathcal A_{\Phi,2\pi,\hbox{\rm SR}}(B)}
\def\beq{\begin{equation}}
\def\ene{\end{equation}}

\def\qed{\ifhmode\unskip\nobreak\fi\ifmmode\ifinner
\else\hskip5pt\fi\fi\hbox{\hskip5pt\vrule width4pt height6pt
depth1.5pt\hskip1pt}}

\def\+out{x^{\rm out}}
\begin{document}

\title{ $L^{p}-L^{p^{\prime}}$ estimates for matrix Schr\"{o}dinger equations
\thanks{ Research partially supported by projects PAPIIT-DGAPA UNAM IN103918, and IA101820, 
and SEP-CONACYT CB 2015, 254062.}}
\author{Ivan Naumkin\thanks{Fellow, Sistema Nacional de Investigadores. Electronic
Mail: ivannaumkinkaikin@gmail.com} and Ricardo Weder\thanks{Fellow, Sistema
Nacional de Investigadores. Electronic mail: weder@unam.mx }\\Departamento de F\'{\i}sica Matem\'{a}tica,\\Instituto de Investigaciones en Matem\'{a}ticas Aplicadas y en Sistemas.\\Universidad Nacional Aut\'{o}noma de M\'{e}xico,\\Apartado Postal 20-126, Ciudad de M\'{e}xico, 01000, M\'{e}xico.}
\date{}
\maketitle

\begin{abstract}
This paper is devoted to the study of dispersive estimates for matrix
Schr\"odinger equations on the half-line with general boundary condition, and
on the line. We prove $L^{p}-L^{p^{\prime}}$ estimates on the half-line for
slowly decaying selfadjoint matrix potentials that satisfy $\int_{0}^{\infty
}\, (1+x) |V(x)|\, dx < \infty$ both in the generic and in the exceptional
cases. We obtain our $L^{p}-L^{p^{\prime}}$ estimate on the line for a $n
\times n$ system, under the condition that $\int_{-^{\infty}}^{\infty}\,
(1+|x|)\, |V(x)|\, dx < \infty,$ from the $L^{p}-L^{p^{\prime}}$ estimate for
a $2n\times2n$ system on the half-line. With our $L^{p}-L^{p^{\prime}}$
estimates we prove Strichartz estimates.

\end{abstract}

\baselineskip=12pt


\section{Introduction.}

In this paper we consider the matrix Schr\"{o}dinger equation on the half-line
with general selfadjoint boundary condition%
\begin{equation}
\left\{
\begin{array}
[c]{c}%
i\partial_{t}u\left(  t,x\right)  =\left(  -\partial_{x}^{2}+V\left(
x\right)  \right)  u\left(  t,x\right)  ,\ t\in\mathbb{R},\,x\in\mathbb{R}%
^{+},\\
u\left(  0,x\right)  =u_{0}\left(  x\right)  ,\text{ }x\in\mathbb{R}^{+},\\
\label{MS}
-B^{\dagger}u\left(  t,0\right)  +A^{\dagger}\left(  \partial_{x}u\right)
\left(  t,0\right)  =0,
\end{array} \right.
\end{equation}
where $\mathbb{R}^{+}:=(0,+\infty),$ $u(t,x)$ is a function from
$\mathbb{R}\times\mathbb{R}^{+}$ into $\mathbb{C}^{n},A,B$ are constant
$n\times n$ matrices, the potential $V$ is a $n\times n$ selfadjoint
matrix-valued function of $x$, i.e.%
\begin{equation}
V\left(  x\right)  =V^{\dagger}\left(  x\right)  ,\text{ }x\in\mathbb{R}^{+},
\label{PotentialHermitian}%
\end{equation}
where the dagger denotes the matrix adjoint. We suppose that $V$ is in the
Faddeev class $L_{1}^{1},$ i.e. that it is a Lebesgue measurable matrix-valued
function and,%

\begin{equation}
\int_{\mathbb R^+}\,(1+x)\,|V(x)|\,dx<\infty, \label{PotentialL11}%
\end{equation}
where by $|V|$ we denote the matrix norm of $V.$ The more general selfadjoint
boundary condition at $x=0$ can be expressed in several equivalent ways
\cite{[5]}, \cite{[6]}, \cite{[9]}, \cite{Harmer1}, \cite{Harmer},
\cite{Harmer2}, \cite{Kostrykin1}, and  \cite{Kostrykin2}. See also  \cite{rofe} for further results on general selfadjoint boundary conditions. We find it convenient to
state the boundary condition following \cite{[5]}, \cite{[6]}, and \cite{[9]}, as
in (\ref{MS}) where the matrices $A$ and $B$ satisfy (see Section
\ref{Selfadjoint realization}),
\begin{equation}
-B^{\dagger}A+A^{\dagger}B=0, \label{wcon1}%
\end{equation}
and%
\begin{equation}
A^{\dagger}A+B^{\dagger}B>0. \label{wcon2}%
\end{equation}
We denote by
\[
H:=H_{A,B,V},
\]
the selfadjoint operator in $L^{2}\left(  \mathbb{R}^{+}\right)  $ associated
to the initial-boundary problem (\ref{MS}).

Currently  there is a considerable interest in  matrix Schr\"{o}dinger
equations. In part, this is because of the importance of these equations for
quantum mechanical scattering of particles with internal structure,  and also since a quantum star graph is a particular case of a matrix Schr\"odinger equation with a diagonal potential matrix.  There is an extensive literature in quantum graphs. See, for example, \cite{Berkolaio}-\cite{Boman},
\cite{Gutkin}, \cite{Kostrykin1}-\cite{Kurasov2}, and the references therein. The matrix Schr\"{o}dinger equation with a diagonal potential matrix corresponds to  a star graph which describes the behavior of $n$ connected very thin quantum wires that form a
star graph, that is, a graph with only one vertex and a finite number of edges
of infinite length. The boundary condition in (\ref{MS}) restrict the value of
the wave function and of its derivative at the vertex. The problem is relevant
from the physical point of view. For instance, it appears in the design of
elementary gates in quantum computing and nanotubes for microscopic electronic
devices, where, for example, strings of atoms may form a star-shaped graph.
The consideration of the most general boundary condition at the vertex and not
only, for say, Dirichlet boundary condition is also physically motivated: for
quantum graphs the relevant boundary conditions are the ones that link the
values, at the different edges, of the wave function and of the first
derivative. An important example is the Kirchoff boundary condition. In fact,
a quantum graph is an idealization of wires with a small cross section that
meet at vertices. It is obtained as the limit when the cross section goes to
zero. The boundary conditions on the graph's vertices depends on how the limit
is taken. In principle, all the boundary conditions in (\ref{MS}) can appear
in this limit procedure. This motivates the study of the more general
selfadjoint boundary condition.

The purpose of this paper is to obtain $L^{p}-L^{p^{\prime}}$ estimates for
the initial- boundary problem (\ref{MS}).

\subsection*{Notation.}

We denote by $L^{p}(U;\mathbb{C}^{n})$, for $1\leq p\leq\infty,$ and
$U={\mathbb{R}}^{+}$ or $U={\mathbb{R}}$, ($\mathbb{R}^{+}$ $=(0,\infty)$),
the standard Lebesgue spaces of $\mathbb{C}^{n}$ valued functions, where
$\mathbb{C}$ denotes the complex numbers. For an integer $m\geq0$ and a real
$1\leq p\leq\infty,$ $W_{m,p}\left(  U;\mathbb{C}^{n}\right)  $ is the
standard Sobolev space. (See e.g.~\cite{AdamsF} for the definitions and
properties of these spaces.) If there is no place for confusion, we shall omit
$\mathbb{C}^{n}$ or both $U\mathbb{\ }$and$\ \mathbb{C}^{n}$ in writing the
above spaces. $W_{m,p}^{\left(  0\right)  }\left(  U\right)  $ is the closure
of $C_{0}^{\infty}\left(  U\right)  $ in the space $W_{m,p}\left(  U\right)
$. We denote the Fourier transform by,
\[
\mathcal{F}f:=\int_{\mathbb{R}}\,e^{ikx}\,f(x)\,dx,
\]
and the inverse Fourier transform by,
\[
\mathcal{F}^{-1}f:=\frac{1}{2\pi}\,\int_{\mathbb{R}}\,e^{-ikx}\,f(k)\,dk.
\]
We designate,
\[
\mathcal{F}\left(  L^{1}(\mathbb{R})\right)  :=\left\{  f\in L^{\infty
}(\mathbb{R}):f=\mathcal{F}g,\text{ }g\in L^{1}(\mathbb{R})\right\}  .
\]
By $\mathbb{C}^{+}$ we denote the open upper-half complex plane. For any pair
of Banach spaces $X,Y$ we denote by $\mathcal{B}(X,Y)$ the Banach space of all
bounded operators from $X$ into $Y.$ When $X=Y$ we use the notation
$\mathcal{B}(X).$ For any operator $G$ in a Banach space $X$ we denote by
$D[G]$ the domain of $G.$ For a bounded below selfadjoint operator, $G,$ the
quadratic form domain of $G$ is the domain of its associated quadratic form
\cite{kato}. By $0_{m}$ and $I_{m},$ $m=1,2,\cdots,$ we designate the $m\times
m$ zero and identity matrices, respectively. Finally, we shall denote by $C$ a
generic positive constant, which does not has to take the same value when it
appears in different places.

\subsection{Main results.}

In order to present our results, let us first define the function spaces we
will work with. We can diagonalize the boundary condition in (\ref{MS}) to get
$n$ equations $\cos\theta_{j}\psi_{j}\left(  0\right)  +\sin\theta_{j}\psi
_{j}^{\prime}\left(  0\right)  =0,0<\theta_{j}\leq\pi,$ $j=1,2,...,n,$ (see
(\ref{PSM13}) below). Let us define the space $\widehat{W}_{j}^{p}\left(
\mathbb{R}^{+}\right)  ,$ for $1\leq p\leq\infty,$ which is the Sobolev space
$W_{1,p}^{\left(  0\right)  }\left(  \mathbb{R}^{+}\right)  ,$ in the case of
Dirichlet boundary condition, $\theta_{j}=\pi,$ and $\widehat{W}_{j}%
^{p}\left(  \mathbb{R}^{+}\right)  =W_{1,p}\left(  \mathbb{R}^{+}\right)  ,$
in the case of Neumann, $\theta_{j}=\pi/2,$ or mixed, $\theta_{j}\neq
\pi,\theta_{j}\neq\pi/2,$ boundary conditions (see (\ref{spaceW1j}) below for
the precise definition). We consider $\widetilde{W}_{1,p}\left(
\mathbb{R}^{+}\right)  =\oplus_{j=1}^{n}\widehat{W}_{j}^{p}\left(
\mathbb{R}^{+}\right)  .$ Then, the quadratic form domain of the Hamiltonian
$H:=H_{A,B}$ that corresponds to the general boundary condition in (\ref{MS})
is given by $W_{1,2}^{A,B}(\mathbb{R}^{+})$ with
\[
W_{1,2}^{A,B}\left(  \mathbb{R}^{+}\right)  =M\widetilde{W}_{1,2}\left(
\mathbb{R}^{+}\right)  \subset W_{1,2}\left(  \mathbb{R}^{+}\right)  ,
\]
where $M$ is a unitary matrix (see (\ref{W1AB}) below). Let
$P_{\operatorname*{c}}$ denote the projector onto the continuous subspace of
$H.$ We observe that $P_{\operatorname*{c}}=I-P_{\operatorname*{p}},$ where
$P_{\operatorname*{p}}$ is the projector onto the subspace of $L^{2}$
generated by the eigenvectors corresponding to the bound states of $H.$ We
also note that under our assumptions on $V,$ the number of negative bound
states of $H$ is finite and that $H$ has no positive or zero bound states.
Hence, the subspace generated by the eigenvectors is finite-dimensional. We
now present our results.

\begin{theorem}
[The $L^{p}-L^{p^{\prime}}$ estimate]\label{Theorem1} Suppose that the
potential $V$ satisfies (\ref{PotentialHermitian}) and (\ref{PotentialL11}).
Then, for any $p\in\lbrack1,2]$ and $p^{\prime}$ such that $1/p+1/p^{\prime
}=1,$ the estimates%
\begin{equation}
\left\Vert e^{-itH}P_{\operatorname*{c}}\right\Vert _{\mathcal{B}\left(
L^{p}(\mathbb{R}^{+}),L^{p^{\prime}}(\mathbb{R}^{+})\right)  } \leq\frac
{C}{\left\vert t\right\vert ^{1/p-1/2}}, \label{estimate1}%
\end{equation}
and%
\begin{equation}
\left\Vert e^{-itH}P_{\operatorname*{c}}\right\Vert _{\mathcal{B}\left(
W_{1,p}^{A,B}(\mathbb{R}^{+}),W_{1,p^{\prime}}^{A,B}(\mathbb{R}^{+})\right)
}\leq\frac{C}{\left\vert t\right\vert ^{1/p-1/2}}, \label{estimate2}%
\end{equation}
hold for all $t\in\mathbb{R\setminus}\{0\}.$
\end{theorem}

\begin{theorem}
[Strichartz estimates]\label{Theorem2}Suppose that the potential $V$ satisfies
(\ref{PotentialHermitian}) and (\ref{PotentialL11}). Let $\left(  q,r\right)
$ be an admissible pair, that is, $2/q=1/2-1/r$ and $2\leq r\leq\infty.$ Then,
for every $\varphi\in L^{2}\left(  \mathbb{R}^{+}\right)  ,$ the function
$t\rightarrow e^{-itH}P_{\operatorname*{c}}\varphi$ belongs to $L^{q}\left(
\mathbb{R}\text{,}L^{r}\left(  \mathbb{R}^{+}\right)  \right)  \cap C\left(
\mathbb{R}\text{,}L^{2}\left(  \mathbb{R}^{+}\right)  \right)  .$ Moreover,
there exists a constant $C>0$ such that
\begin{equation}
\left\Vert e^{-itH}P_{\operatorname*{c}}\varphi\right\Vert _{L^{q}\left(
\mathbb{R}\text{,}L^{r}\left(  \mathbb{R}^{+}\right)  \right)  }\leq
C\left\Vert \varphi\right\Vert _{L^{2}\left(  \mathbb{R}^{+}\right)  },
\label{Strichartz estimates}%
\end{equation}
for every $\varphi\in L^{2}\left(  \mathbb{R}^{+} \right)  .$ Moreover, let $I
\subset\mathbb{R}$ be an interval. For an admissible pair $\left(  \gamma
,\rho\right)  ,$ let $f\in L^{\gamma^{\prime}}\left(  I\text{,}L^{\rho
^{\prime}}\left(  \mathbb{R}^{+}\right)  \right)  ,$ where $1/\gamma
+1/\gamma^{\prime}=1$ and $1/\rho+1/\rho^{\prime}=1.$ Then, for $t_{0}\in
\bar{I},$ the function
\[
t\rightarrow\Phi_{f}\left(  t\right)  =\int_{t_{0}}^{t}e^{-i(t-s) H}%
P_{\operatorname*{c}}f\left(  s\right)  ds,\text{ }t\in I,
\]
belongs to $L^{q}\left(  I\text{,}L^{r}\left(  \mathbb{R}^{+}\right)  \right)
\cap C\left(  \bar{I}\text{,}L^{2}\left(  \mathbb{R}^{+}\right)  \right)  $
and
\[
\left\Vert \Phi_{f}\right\Vert _{L^{q}\left(  I,L^{r}\left(  \mathbb{R}%
^{+}\right)  \right)  }\leq C\left\Vert f\right\Vert _{L^{\gamma^{\prime}%
}\left(  I,L^{\rho^{\prime}}\left(  \mathbb{R}^{+}\right)  \right)  },\text{
for every }f\in L^{\gamma^{\prime}}\left(  I,L^{\rho^{\prime}}\left(
\mathbb{R}^{+}\right)  \right)  ,
\]
where the constant $C$ is independent of $I.$
\end{theorem}


\subsubsection{The matrix Schr\"{o}dinger equation on the 
full-line.
\label{Schrodinger on the line}}

Following \cite{WederBook} we show that a $2n\times2n$ matrix Schr\"{o}dinger
equation on the half-line is unitarily equivalent to a $n\times n$ matrix
Schr\"{o}dinger equation on the full-line with a point interaction at $x=0.$
We define the unitary operator $\mathbf{U}$ from $L^{2}\left(  \mathbb{R}%
^{+};\mathbb{C}^{2n}\right)  $ onto $L^{2}\left(  \mathbb{R};\mathbb{C}%
^{n}\right)  $ by%
\begin{equation}
\phi\left(  x\right)  =\mathbf{U }\psi\left(  x\right)  :=\left\{
\begin{array}
[c]{c}%
\psi_{1}\left(  x\right)  ,\text{ \ }x\geq0,\\
\psi_{2}\left(  -x\right)  ,\text{ \ }x<0,
\end{array}
\right.  \label{unitarytransform}%
\end{equation}
for a vector-valued function $\psi=\left(  \psi_{1},\psi_{2}\right)  ^{T},$
($T$ denotes the matrix transpose) where $\psi_{j}\in L^{2}\left(
\mathbb{R}^{+};\mathbb{C}^{n}\right)  ,$ $j=1,2.$ Let the potential in
(\ref{MS}) be the diagonal matrix
\[
V\left(  x\right)  :=\operatorname*{diag}\{V_{1}\left(  x\right)
,V_{2}\left(  x\right)  \},
\]
where $V_{j}, j=1,2$ are selfadjoint $n \times n$ matrix-valued functions that
satisfy $V_{j} \in L^{1}_{1}(\mathbb{R}^{+}).$ Under $\mathbf{U}$ the
Hamiltonian $H$ is transformed into the following Hamiltonian in the full-line,%

\begin{equation}
H_{\mathbb{R}}:=\mathbf{U}\,H\mathbf{U}^{\dagger},\quad D[H_{\mathbb{R}%
}]:=\{\phi\in L^{2}\left(  \mathbb{R};\mathbb{C}^{n}\right)  :\mathbf{U}%
^{\dagger}\phi\in D[H]\}. \label{hfull}%
\end{equation}
The operator $H_{\mathbb{R}}$ is a selfadjoint realization in $L^{2}\left(
\mathbb{R};\mathbb{C}^{n}\right)  $ of the formal differential operator
$-\partial_{x}^{2}+Q(x)$ where,
\[
Q\left(  x\right)  =\left\{
\begin{array}
[c]{c}%
V_{1}\left(  x\right)  ,\text{ \ }x\geq0,\\
V_{2}\left(  -x\right)  ,\text{ \ }x<0.
\end{array}
\right.
\]
Further, the quadratic form domain of $H_{\mathbb{R}}$ is given by
$W_{1,2}^{\mathbb{R},A,B}$ where,
\[
W_{1,2}^{\mathbb{R},A,B}:=\mathbf{U}W_{1,2}^{A,B}\subset W_{1,2}%
(-\infty,0)\oplus W_{1,2}(0,\infty).
\]
Let us write the $2n\times2n$ matrices $A,B$ as follows,
\begin{equation}
A=\left[
\begin{array}
[c]{l}%
A_{1}\\
A_{2}%
\end{array}
\right]  ,\quad B=  \left[
\begin{array}
[c]{l}%
B_{1}\\
B_{2}%
\end{array}
\right]  , \label{matrices}%
\end{equation}
with $A_{j},B_{j},j=1,2,$ being $n\times2n$ matrices. We have that the
functions in the domain of $H_{\mathbb{R}}$ satisfy the following transmission
condition at $x=0,$%

\begin{equation}
-B_{1}^{\dagger}\phi(0+)-B_{2}^{\dagger}\phi(0-)+A_{1}^{\dagger}(\partial
_{x}\phi)(0+)-A_{2}^{\dagger}(\partial_{x}\phi)(0-)=0. \label{bdcond}%
\end{equation}
Then, $u(t,x)$ is a solution of the problem (\ref{MS}) if and only if
$v(t,x):=\mathbf{U}u(t,x)$ is a solution of the following $n\times n$ system
in the full-line,
\begin{equation}
\left\{
\begin{array}
[c]{c}%
i\partial_{t}v\left(  t,x\right)  =\left(  -\partial_{x}^{2}+Q\left(
x\right)  \right)  v\left(  t,x\right)  ,\ t\in\mathbb{R},\,x\in\mathbb{R},\\
v\left(  0,x\right)  =v_{0}\left(  x\right)  :=\mathbf{U}u_{0}\left(
x\right)  ,x\in\mathbb{R},\\
-B_{1}^{\dagger}v(t,0+)-B_{2}^{\dagger}v(t,0-)+A_{1}^{\dagger}(\partial
_{x}v)(t,0+)-A_{2}^{\dagger}(\partial_{x}v)(t,0-)=0.
\end{array}
\right.  \label{MSR}%
\end{equation}
For example, let us take,
\[
A=\left[
\begin{array}
[c]{lc}%
0_{n} & I_{n}\\
0_{n} & I_{n}%
\end{array}
\right]  ,\quad B=\left[
\begin{array}
[c]{lc}%
-I_{n} & \Lambda\\
\ I_{n} & 0_{n}%
\end{array}
\right]  ,
\]
where $\Lambda$ is a selfadjoint $n\times n$ matrix. These matrices satisfy
(\ref{wcon1}, \ref{wcon2}). Moreover, the transmission condition in
(\ref{MSR}) is given by,
\begin{equation}
v(t,0+)=v(t,0-)=v(t,0),\quad(\partial_{x}v)(t,0+)-(\partial_{x}%
v)(t,0-)=\Lambda v(t,0). \label{bcond2}%
\end{equation}
This transmission condition corresponds to a Dirac delta point interaction at
$x=0$ with coupling matrix $\Lambda$. If $\Lambda=0,$ $v(t,x)$ and
$(\partial_{x}v)(t,x)$ are continuous at $x=0$ and the transmission condition
corresponds to the matrix Schr\"{o}dinger equation on the full-line without a
point interaction at $x=0.$

Using Theorem \ref{Theorem2} and the unitary operator $\mathbf{U},$ as above,
we deduce the following result concerning the Cauchy problem (\ref{MSR}).

\begin{corollary}
\label{wholeline} (The full-line case)\label{CorollaryWhole} Let
$n\in\mathbb{N}$. Suppose that $Q\left(  x\right)  ,$ $x\in\mathbb{R}$, is a
$n\times n$ selfadjoint matrix-valued function such that $Q\in L_{1}%
^{1}\left(  \mathbb{R};\mathbb{C}^{n}\right)  .$ Then, for any $p\in
\lbrack1,2]$ and $p^{\prime}$ such that $1/p+1/p^{\prime}=1,$ the estimates%
\[
\left\Vert e^{-itH_{\mathbb{R}}}P_{\operatorname*{c},\mathbb{R}}\right\Vert
_{\mathcal{B}\left(  L^{p}\left(  \mathbb{R};\mathbb{C}^{n}\right)
,L^{p^{\prime}}\left(  \mathbb{R};\mathbb{C}^{n}\right)  \right)  }\leq
\frac{C}{\left\vert t\right\vert ^{1/p-1/2}},
\]
and%
\[
\left\Vert e^{-itH_{\mathbb{R}}}P_{\operatorname*{c},\mathbb{R}}\right\Vert
_{\mathcal{B}\left(  W_{1,p}^{\mathbb{R},A,B},W_{1,p^{\prime}}^{\mathbb{R}%
,A,B}\right)  }\leq\frac{C}{\left\vert t\right\vert ^{1/p-1/2}},
\]
hold for all $t\in\mathbb{R\setminus}\{0\},$ where $P_{\operatorname*{c}%
,\mathbb{R}}$ is the projector onto the continuous subspace of $H_{\mathbb{R}%
}.$ Moreover, let $\left(  q,r\right)  $ be an admissible pair, that is,
$2/q=1/2-1/r$ and $2\leq r\leq\infty.$ Then, the conclusions of Theorem
\ref{Theorem2} are true with $\mathbb{R}^{+}$ replaced by $\mathbb{R}$ and
with $H_{\mathbb{R}},P_{\operatorname*{c},\mathbb{R}}$ instead, respectively,
of $H$ and $P_{\operatorname*{c}}.$
\end{corollary}

\subsection*{Comments on the results and on the literature.}
In the  case of star graphs with  potential  identically zero, and  with general boundary conditions,  $L^p-L^{p'}$ estimates, and  Strichartz estimates were obtained by \cite{grecu}. Moreover, for a star graph with the Kirchoff boundary condition and a potential that satisfies $\int_0^\infty\, dx   (1+x)^\gamma |V(x)| < \infty, \gamma > 5/2,$ $L^p-L^{p'}$ estimates, and  Strichartz estimates were  proven in \cite{man}.
Note that Theorems \ref{Theorem1} and \ref{Theorem2} and Corollary
\ref{wholeline} hold under the same conditions in the \textit{generic} and
\textit{exceptional} cases. Recall that we are in the generic case if the Jost
matrix is invertible at zero energy and that we are in the exceptional case if
the Jost matrix is not invertible at zero energy. In the exceptional case
there is a resonance (or half-bound state) with zero energy, and in the
generic case there is no resonance  at zero energy. In
other words, the validity of the dispersive estimates is independent of the
existence of a resonance  with zero energy.

In order to obtain the $L^{p}-L^{p^{\prime}}$ estimates, we follow the
approach of \cite{WederBoundary}. For this purpose, we use the scattering
theory for the matrix Schr\"{o}dinger equation on the half-line developed in
\cite{AgrMarch}, \cite{[5]}, \cite{[9]}, \cite{akwearxiv},\cite{WederBook} and
\cite{WederStar}. From the spectral representation for the matrix
Schr\"odinger operator $H$ we get a representation (see (\ref{eq4}) below) for
the continuous part (which corresponds to the scattering process) of the
evolution group $e^{-itH}$ in terms of the Jost solutions for the stationary
matrix Schr\"{o}dinger equation. Then, we can estimate the large-time behavior
of $e^{-itH}P_{\operatorname*{c}}$ by using the low- and high-energy
behaviours of the Jost solutions and the scattering matrix. For this purpose,
we need to estimate the difference between the scattering matrix and its
high-energy limit and to show that the Fourier transform of the difference is
integrable on the whole real line. This is Theorem \ref{FourierL1} below. This
result, which is interesting by its own, is crucial for obtaining the
$L^{p}-L^{p^{\prime}}$ estimates for such general perturbations as $V\in
L_{1}^{1}\left(  \mathbb{R}^{+}\right)  .$ We prove Theorem \ref{FourierL1} by
adapting the arguments of \cite{AgrMarch} for the Dirichlet boundary
condition, which involve the well-known Wiener theorem, to the case of general
self-adjoint boundary condition in (\ref{MS}). The key technical tools that
allows us to prove that the Fourier transform of the scattering matrix minus
its high-energy limit is integrable, under this generality, are the sharp
results on the low-energy behavior of the Jost matrix, including a formula for
the Jost matrix at zero energy, that where obtained in \cite{[5]} and the
precise estimate of the high-energy behavior of the scattering matrix of
\cite{[9]}. We observe that an alternative method for obtaining the $L^{p}
-L^{p\prime}$ estimates is developed in \cite{WederWholeLine}. This approach
requires a more detailed and subtle study of the low-energy properties of the
scattering data. Hence, it needs stronger conditions.

There is a very extensive literature on dispersive estimates. For surveys see
\cite{Fanelli} and \cite{Schlag}. We will only comment on results in one
dimension. The $L^{p}-L^{p\prime}$ estimates on the line were first proven in
the scalar case by Weder \cite{WederWholeLine} under the condition
\begin{equation}
\int_{\mathbb{R}}\,(1+|x|)^{\gamma}\,|V(x)|\,dx<\infty, \label{condition}%
\end{equation}
with $\gamma>3/2$ in the generic case and $\gamma>5/2$ in the exceptional
case. This was generalized by M. Goldberg and W. Schlag \cite{goldschl} to,
respectively $\gamma=1$ and $\gamma=2,$ and by Egorova, Kopylova, Marchenko, and Teschl
\cite{kmt} to $\gamma=1$ in the generic and the exceptional cases. D'Ancona
and Selberg \cite{as} considered a potential that satisfies (\ref{condition})
with $\gamma=2$ plus a step potential. Note that Corollary \ref{wholeline}
with the point interaction at $x=0$ is new in the scalar case. We are not
aware of any result on $L^{p}-L^{p\prime}$ estimates on the line for matrix
Schr\"{o}dinger equations.

The $L^{p}-L^{p\prime}$ estimates on the half-line, in the scalar case and
with Dirichlet boundary condition was proven by Weder \cite{WederBoundary}
under the condition $\int_{0}^{\infty}\,x\,|V(x)|\,dx<\infty$ in the generic
and the exceptional cases. It was actually in this paper that it was
discovered that the $L^{p}-L^{p\prime}$ estimates hold under the same
condition in the generic and the exceptional cases. The case of the spherical
Schr\"{o}dinger equation was considered by Holzleitner, Kostenko and Teschl
\cite{hkt2} and by Kostenko, Teschl and Toloza \cite{ktt}. The case of the one-dimensional Klein-Gordon
equation with a potential was studied by Weder \cite{wkg},  Egorova,  Kopylova,
Marchenko, and Teschl \cite{kmt} and by Prill \cite{op}. Kopylova and Teschl
\cite{kt} considered one dimensional discrete Dirac equations.

\bigskip

The paper is organized as follows.\ In Section \ref{Section2} we consider
results concerning the scattering theory for the matrix Schr\"{o}dinger
equation on the half-line, which play a crucial role in the proof of our
dispersive estimates. In particular, in Subsection
\ref{Selfadjoint realization} we construct the self-adjoint extension $H$
associated to the matrix Schr\"{o}dinger equation (\ref{MS}). In Subsection
\ref{Scattering Data} we introduce the relevant solutions for the stationary
matrix Schr\"{o}dinger equation. Using these solutions, in Subsection
\ref{Generalized Fourier transform}, we construct the spectral representations
for the operator $H$ via the generalized Fourier transforms. In Section
\ref{ScatteringFourier} we prove that the Fourier transform of the scattering
matrix minus its high-energy limit is integrable on the line. We use the
results of Section \ref{Section2} in Section \ref{Lp-Lq} to prove the
$L^{p}-L^{p^{\prime}}$ and Strichartz estimates for the matrix Schr\"{o}dinger equation.

\section{Scattering for Matrix Schr\"{o}dinger Equations.\label{Section2}}

\subsection{The Schr\"{o}dinger equation on the
half-line\label{Selfadjoint realization}.}

Let $n\in\mathbb{N}$. Consider the stationary matrix Schr\"{o}dinger equation
on the half-line%
\begin{equation}
-\psi^{\prime\prime}+V\left(  x\right)  \psi=k^{2}\psi,\text{ }x\in
\mathbb{R}^{+}, \label{MSStationary}%
\end{equation}
where the prime denotes the derivative with respect to the spatial coordinate
$x$, $k^{2}$ is the complex-valued spectral parameter, $V(x)$ satisfies
(\ref{PotentialHermitian}) and is such that
\begin{equation}
V\in L^{1}\left(  \mathbb{R}^{+}\right)  . \label{PotentialL1}%
\end{equation}
The wavefunction $\psi(k,x)$ appearing may be either a $n\times n$
matrix-valued function or it may be a column vector with $n$ components. As
mentioned at the beginning of the introduction, the more general selfadjoint
boundary condition at $x=0$ can be expressed in terms of two constant $n\times
n$ matrices $A$ and $B$ as%

\begin{equation}
-B^{\dagger}\psi\left(  0\right)  +A^{\dagger}\psi^{\prime}\left(  0\right)
=0, \label{SA1}%
\end{equation}
where $A$ and $B$ satisfy
\begin{equation}
-B^{\dagger}A+A^{\dagger}B=0, \label{SA2}%
\end{equation}%
\begin{equation}
A^{\dagger}A+B^{\dagger}B>0. \label{SA3}%
\end{equation}

We observe that \cite{[9]} provides the explicit steps to go from any pair of
matrices $A$ and $B$ appearing in the selfadjoint boundary condition
(\ref{SA1})-(\ref{SA3}) to a pair $\tilde{A}$ and $\tilde{B},$ given by%
\begin{equation}
\tilde{A}=-\operatorname*{diag}[\sin\theta_{1},...,\sin\theta_{n}],\text{
}\tilde{B}=\operatorname*{diag}[\cos\theta_{1},...,\cos\theta_{n}],
\label{A,Btilde}%
\end{equation}
with appropriate real parameters $\theta_{j}\in(0,\pi],$ which still satisfy
(\ref{SA1})-(\ref{SA3}). For the matrices $\tilde{A}$, $\tilde{B},$ the
boundary conditions (\ref{SA1}) are given by%
\begin{equation}
\cos\theta_{j}\psi_{j}\left(  0\right)  +\sin\theta_{j}\psi_{j}^{\prime
}\left(  0\right)  =0,\text{ \ \ }j=1,2,...,n, \label{PSM13}%
\end{equation}
 The special case $\theta_{j}=\pi$
corresponds to the Dirichlet boundary condition and the case $\theta_{j}%
=\pi/2$ corresponds to the Neumann boundary condition. In general, there are
$n_{\operatorname*{N}}\leq n$ values with $\theta_{j}=\pi/2$ and
$n_{\operatorname*{D}}\leq n$ values with $\theta_{j}=\pi$, and hence there
are $n_{\operatorname*{M}}$ remaining values, with $n_{\operatorname*{M}%
}=n-n_{\operatorname*{N}}-n_{\operatorname*{D}}$ such that those $\theta_{j}%
$-values lie in the interval $(0,\pi/2)$ or $(\pi/2,\pi),$ i.e., they
correspond to mixed boundary conditions. 
 In fact, it is proven in \cite{[9]} that for any
pair of matrices $(A,B)$ that satisfy (\ref{SA1}, \ref{SA2}) there is a pair
of matrices $(\tilde{A},\tilde{B})$ as in (\ref{A,Btilde}), a unitary matrix
$M$ and two invertible matrices $T_{1},T_{2}$ such%

\begin{equation}
A=M\,\tilde{A}T_{1}M^{\dagger}T_{2},\quad B=M\,\tilde{B}T_{1}M^{\dagger}T_{2}.
\label{trans}%
\end{equation}
We construct a selfadjoint realization of the matrix Schr\"{o}dinger operator
$-\partial_{x}^{2}+V(x)$ by quadratic forms methods. For the following
discussion see \cite{WederBook} and \cite{WederStar}. Let $\theta_{j}$ be given by equations
(\ref{PSM13}). For $1\leq p\leq\infty,$ we denote
\begin{equation}
\widehat{W}_{j}^{p}:=W_{1,p}^{\left(  0\right)  },\text{ if }\theta_{j}%
=\pi,\text{ and }\widehat{W}_{j}^{p}:=W_{1,p},\text{ if }\theta_{j}\neq\pi.
\label{spaceW1j}%
\end{equation}
We put%
\[
\widetilde{W}_{1,p}:=\oplus_{j=1}^{n}\widehat{W}_{j}^{p}.
\]
We write%
\[
\Theta:=\operatorname*{diag}[\widehat{\cot}\theta_{1},...,\widehat{\cot}%
\theta_{n}],
\]
where $\widehat{\cot}\theta_{j}=0,$ if $\theta_{j}=\pi/2,$ or $\theta_{j}%
=\pi,$ and $\widehat{\cot}\theta_{j}=\cot\theta_{j},$ if $\theta_{j}\neq
\pi/2,\pi.$ Suppose that the potential $V$ satisfies (\ref{PotentialHermitian}%
) and (\ref{PotentialL1}).
The following quadratic form is closed, symmetric and bounded below,
\begin{equation}
h\left(  \phi,\psi\right)  :=\left(  \phi^{\prime},\psi^{\prime}\right)
_{L^{2}}-\left\langle M\Theta M^{\dag}\phi\left(  0\right)  ,\psi\left(
0\right)  \right\rangle +\left(  V\phi,\psi\right)  _{L^{2}},\text{ }Q\left(
h\right)  :=W_{1,2}^{A,B}, \label{quadratic}%
\end{equation}
where by $Q(h)$ we denote the domain of $h$ and,
\begin{equation}
W_{1,p}^{A,B}:=M\widetilde{W}_{1,p}\subset W_{1,p}. \label{W1AB}%
\end{equation}
Further, by $\langle\cdot,\cdot\rangle$ we designate the scalar product in
$\mathbb{C}^{n}.$ We denote by $H_{A,B,V}$ the selfadjoint bounded below
operator associated to $h$ \cite{kato}. The operator $H_{A,B,V}$ is the
selfadjoint realization of $-\partial_{x}^{2}+V\left(  x\right)  $ with the
selfadjoint boundary condition (\ref{SA1}). When there is no possibility of
misunderstanding we will use the notation $H,$ i.e., $H\equiv H_{A,B,V}.$ It
is proven in \cite{WederBook} and \cite{WederStar} that,
\begin{equation}
H_{A,B,V}=MH_{\tilde{A},\tilde{B},M^{\dagger}VM}M^{\dagger}.
\label{diagonalization}%
\end{equation}
We denote by $H_{D,N,0}$ the selfadjoint bounded below operator associated to
the quadratic form (\ref{quadratic}) with $V\equiv0$ and the $\theta_{j}$
corresponding to the $n_{M}$ mixed boundary conditions replaced by $\theta
_{j}=\pi/2,$ i.e. with the mixed boundary conditions replaced by Neumann
boundary conditions. Note that the quadratic form domain of $H_{D,N,0}$ is
$W_{1,2}^{A,B}.$ Take $L>1$ such that $H+L>I$ and $H_{D,N,0}+L>I.$ Hence,
since the domains of $\sqrt{H+L}$ and of $\sqrt{H_{D,N,0}+L}$ are equal to
$W_{1,2}^{A,B}$ we have that,
\begin{equation}
\left(  \sqrt{H+L}\right)  \,\left(  \sqrt{H_{D,N,0}+L}\right)  ^{-1}%
\in\mathcal{B}(L^{2}),\quad\left(  \sqrt{H_{D,N,0}+L}\right)  \,\left(
\sqrt{H+L}\right)  ^{-1}\in\mathcal{B}(L^{2}). \label{equiv}%
\end{equation}
Denote by $\mathcal{H}$ the domain of $\sqrt{H+L}$ endowed with the norm,
\begin{equation}
\Vert\phi\Vert_{\mathcal{H}}:=\Vert\sqrt{H+L}\phi\Vert_{L^{2}},\quad\phi
\in\mathcal{H}. \label{norm}%
\end{equation}
In other words, $\mathcal{H}$ consists of $W_{1,2}^{A,B},$ but with the norm
(\ref{norm}). Observe that it follows from (\ref{quadratic}) that,
\[
\Vert\phi\Vert_{\mathcal{H}}^{2}=\left(  \phi^{\prime},\phi^{\prime}\right)
_{L^{2}}-\left\langle M\Theta M^{\dag}\phi\left(  0\right)  ,\phi\left(
0\right)  \right\rangle +\left(  V\phi,\phi\right)  _{L^{2}}+L\left(
\phi,\phi\right)  _{L^{2}},\quad\phi\in\mathcal{H}.
\]
Similarly,
\[
\Vert\sqrt{H_{D,N,0}+L}\phi\Vert_{L^{2}}^{2}=\left(  \phi^{\prime}%
,\phi^{\prime}\right)  _{L^{2}}+L\left(  \phi,\phi\right)  _{L^{2}}=\Vert
\phi\Vert_{W_{1,2}^{A,B}}^{2}+(L-1)\Vert\phi\Vert_{L^{2}}^{2},\quad\phi\in
W_{1,2}^{A,B}.
\]
Moreover, by (\ref{equiv}) there are positive constants $C_{1},C_{2}$ such
that,
\begin{equation}
C_{1}\Vert\phi\Vert_{W_{1,2}^{A,B}}\leq\Vert\phi\Vert_{\mathcal{H}}\leq
C_{2}\Vert\phi\Vert_{W_{1,2}^{A,B}}. \label{equiv2}%
\end{equation}

\subsection{The Jost and scattering matrices.\label{Scattering Data}}
Below, in Propositions~\ref{PJostsolution}, \ref{Jostnozero} and \ref{PK} we state results in special solutions to the matrix Schr\"odinger equation that we use. The interested reader can consult the monographs, \cite{[14],[37], [38]}, and the references quoted there, for similar results in the scalar case. 

We now introduce the special solutions for (\ref{MSStationary}) that play a
crucial role in our analysis.
By \cite{AgrMarch},
\cite{[5]}, \cite{WederBook} we have that:

\begin{proposition}
\label{PJostsolution}Suppose that the potential $V$ satisfies
(\ref{PotentialL1}). For each fixed
$k\in\overline{\mathbb{C}^{+}}\backslash\{0\}$ there exists a unique $n\times
n$ matrix-valued Jost solution $f\left(  k,x\right)  $ to equation
(\ref{MSStationary}) satisfying the asymptotic condition%
\begin{equation}
f\left(  k,x\right)  =e^{ikx}\left(  I+o\left( 1 \right)\right), \quad f'(k,x)= e^{ikx}[ik\, I+ o(1)], \qquad  x \to + \infty. \label{Jostsolution}%
\end{equation}
For each $k\in \overline{\mathbf C^+}\setminus\{0\},$
the quantity $f(k,x)$ and its $x$-derivative
$f'(k,x)$
are continuous in $x\in[0,+\infty).$
Moreover, for any fixed $x\in [0,\infty),$ $f\left(  k,x\right) $   and   $f'\left(  k,x\right) $  are
analytic in $k\in\mathbb{C}^{+}$ and continuous in $k\in\overline
{\mathbb{C}^{+}} \setminus\{0\}$.
\end{proposition}

Note that  \cite{AgrMarch} proves a result that is slightly different from the one given in  Proposition ~\ref{PJostsolution}, because they use Jost solutions analytic in
 $\mathbb{C}^{-}$. Furthermore, \cite{AgrMarch} considers potentials such that $x^{1+\delta} V(x)$ is integrable for $ \delta \geq 0,$ but they obtain a sharper error bound that depends on $\delta.$   Proposition ~\ref{PJostsolution}, as we state it above,  for Jost solutions analytic in  $\mathbb{C}^{+},$  and for potentials that satisfy (\ref{PotentialL1}) is given    in \cite{[5]}, \cite{WederBook}.

Given the boundary matrices $A$ and $B$ satisfying (\ref{SA2})-(\ref{SA3}),
from the Jost solution we construct the Jost matrix $J\left(  k\right)  ,$
which is a $n\times n$ matrix-valued function of $k,$
\begin{equation}
J\left(  k\right)  =f\left(  -k^{\ast},0\right)  ^{\dagger}B-f^{\prime}\left(
-k^{\ast},0\right)  ^{\dagger}A,\text{ \ }k\in\overline{\mathbb{C}^{+}},
\label{Jostmatrix}%
\end{equation}
where the asterisk denotes complex conjugation. The following proposition is
proven in \cite{AgrMarch}, \cite{[5]}, \cite{WederBook}, and \cite{Harmer}.

\begin{proposition}
\label{Jostnozero}Suppose that the potential $V$ satisfies
(\ref{PotentialHermitian}) and \eqref{PotentialL1}. Then, the Jost matrix
$J\left(  k\right)  $ is analytic for $k\in\mathbb{C}^{+}$, continuous for
$k\in\overline{\mathbb{C}^{+}}\setminus\{0\}$ and invertible for $k\in\mathbb{R}%
\setminus\{0\}.$ If furthermore, the potential satisfies \eqref{PotentialL11}, then, the Jost matrix is continuos for $k\in\overline{\mathbb{C}^{+}}.$ 
\end{proposition}

Note that  \cite{AgrMarch} states a result in the case of Dirichlet boundary condition that is slightly different from the one given in  Proposition ~ \ref{Jostnozero}, because they use Jost solutions analytic in $\mathbb{C}^{-},$ moreover, they always assume that $ x V(x)$ is integrable.  Proposition ~\ref{Jostnozero} as we state it, for Jost solutions analytic in $\mathbb{C}^{+},$  and   for general boundary condition is given in   \cite{[5]}, \cite{WederBook}, and  \cite{Harmer}. In \cite{Harmer} it is always assumed that 
\eqref{PotentialL11} holds.

For $x\geq0,$ let $K\left(  x,y\right)  $ be defined as
\[
K\left(  x,y\right)  =\left(  2\pi\right)  ^{-1}\int_{-\infty}^{\infty
}[f\left(  k,x\right)  -e^{ikx}I]e^{-iky}dk.
\]
Let us define the functions
\[
\sigma\left(  x\right)  =\int_{x}^{\infty}\left\vert V\left(  y\right)
\right\vert dy,\text{ \ }\sigma_{1}\left(  x\right)  =\int_{x}^{\infty
}y\left\vert V\left(  y\right)  \right\vert dy,\text{ }x\geq0.
\]
We observe that for potentials satisfying (\ref{PotentialL11}), both
$\sigma\left(  0\right)  $ and $\sigma_{1}\left(  0\right)  $ are finite, and
moreover,
\[
\int_{0}^{\infty}\,\sigma(x)\,dx=\sigma_{1}(0)<\infty.
\]
The following proposition is given in \cite{AgrMarch}.

\begin{proposition}
\label{PK}Suppose that the potential $V$ satisfies (\ref{PotentialHermitian})
and (\ref{PotentialL11}). Then, the matrix $K\left(  x,y\right)  $ is
continuous in $(x,y)$ in the region $0\leq x\leq y ,$ and is
related to the potential via
\[
K\left(  x,x^+\right)  =\frac{1}{2}\int_{x}^{\infty}V\left(  z\right)  dz,\text{
\ }x\in\lbrack0,+\infty).
\]
The Jost solution $f\left(  k,x\right)  $ has the representation
\begin{equation}
f\left(  k,x\right)  =e^{ikx}I+\int_{x}^{\infty}e^{iky}K\left(  x,y\right)
dy. \label{jostK}%
\end{equation}
The matrix $K\left(  x,y\right)  $ satisfies,
\begin{align}
K\left(  x,y\right)   &  =0,\text{ }y<x, x,y \in [0,\infty),  \nonumber\\
\left\vert K\left(  x,y\right)  \right\vert  &  \leq\frac{1}{2}e^{\sigma
_{1}\left(  x\right)  }\sigma\left(  \frac{x+y}{2}\right)  ,\text{\ }%
x,y \in [0,\infty),\label{ESTK}%
\end{align}%
\begin{align}
\partial_{x}^{i}\partial_{y}^{j}K\left(  x,y\right)   &  =0,\text{
}y<x,  x,y \in [0,\infty),     \nonumber\\
\left\vert \partial_{x}^{i}\partial_{y}^{j}K\left(  x,y\right)  \right\vert
&  \leq\frac{1}{4}\left\vert V\left(  \frac{x+y}{2}\right)  \right\vert
+\frac{1}{2}e^{\sigma_{1}\left(  x\right)  }\sigma\left(  \frac{x+y}%
{2}\right)  \sigma\left(  x\right)  ,\text{ }0<x<y,\text{ }i+j=1.
\label{ESTK1}%
\end{align}

\end{proposition}
Note that  \cite{AgrMarch} states the result  in a  slightly different form from the one in  Proposition ~ \ref{PK}, because they use Jost solutions analytic in $\mathbb{C}^{-}.$

We observe that the Jost matrix $J\left(  k\right)  $ can be expressed in
terms of $K$ as%

\begin{equation}
J\left(  k\right)  ^{\dagger}=B^{\dagger}+ikA^{\dagger}+A^{\dagger}K\left(
0,0\right)  +\int_{0}^{\infty}e^{-ikz}\left(  B^{\dagger}K\left(  0,z\right)
-A^{\dagger}K_{x}\left(  0,z\right)  \right)  dz,\text{ }k\in\mathbb{R}%
\text{.} \label{Josttranspose}%
\end{equation}
From the Jost matrix $J\left(  k\right)  $ we construct the scattering matrix
$S\left(  k\right)  ,$ which is a $n\times n$ matrix-valued function of $k$
given by
\begin{equation}
S\left(  k\right)  =-J\left(  -k\right)  J\left(  k\right)  ^{-1},\text{ }%
k\in\mathbb{R}\text{.} \label{Scatteringmatrix}%
\end{equation}
In the exceptional case where $J(0)$ is not invertible the scattering matrix is defined  by \eqref{Scatteringmatrix} only for $ k \neq 0.$ However, it is proven in \cite{[5]} that for potentials satisfying (\ref{PotentialHermitian}) and (\ref{PotentialL11})
the limit $S(0):=\lim_{k\to 0} S(k)$ exists in the exceptional case and, moreover,  a formula for $S(0)$ is given.  Actually, the low-energy analysis of \cite{[5]} plays a crucial role in the proof of Theorem~\ref{FourierL1}. Further, it is proven  in \cite{[5]} that the relation 

\begin{equation}
S\left(  -k\right)  =S\left(  k\right)  ^{\dagger}=S\left(  k\right)
^{-1},\text{ }k\in\mathbb{R}\text{,} \label{UnitaritySM}%
\end{equation}
holds. In particular, the scattering matrix $S\left(  k\right)  $
is unitary for $k\in\mathbb{R}$ and
\begin{equation}
S\left(  -k\right)  =-\left(  J\left(  k\right)  ^{\dagger}\right)
^{-1}\left(  J\left(  -k\right)  \right)  ^{\dagger}.
\label{Scatteringmatrix1}%
\end{equation}
In terms of the Jost solution $f(k,x)$ and the scattering matrix $S(k)$ we
construct the physical solution \cite{[9]}%
\begin{equation}
\Psi\left(  k,x\right)  =f\left(  -k,x\right)  +f\left(  k,x\right)  S\left(k\right)  ,k\in\mathbb{R}. \label{Physicalsolution}%
\end{equation}
For the definition of the physical solution in the scalar case the reader can consult \cite{[14],[37], [38]}. Observe that \cite{AgrMarch} gives a definition of the physical solution in the case with Dirichlet boundary condition that is different from \eqref{Physicalsolution}. Recall that they use Jost solutions analytic in $\mathbb C^{-}.$
Observe that by (\ref{jostK}), (\ref{ESTK}) and the unitarity of $S(k)$ we
have
\begin{equation}
\left\vert \Psi\left(  k,x\right)  \right\vert \leq C. \label{Psi1}%
\end{equation}
The physical solution $\Psi$ is the basis to construct the generalized Fourier
maps for the absolutely continuous subspace of $H.$ We observe that in the
case when $V=0,$ $f\left(  k,x\right)  =e^{ikx}I$ and then, it follows from
(\ref{Jostmatrix}) and (\ref{Scatteringmatrix}) that%
\begin{align}
J_{0}\left(  k\right)   &  =B-ikA,\text{ \ }J_{0}^{-1}\left(  k\right)
=\left(  B-ikA\right)  ^{-1},\label{J0}\\
S_{0}\left(  k\right)   &  =-\left(  B+ikA\right)  \left(  B-ikA\right)
^{-1}, \label{w00}%
\end{align}
where the zero index refers to the zero potential. In the diagonal form
$\tilde{A}$ and $\tilde{B}$ given by (\ref{A,Btilde}), the Jost and the
scattering matrices take the form%
\begin{equation}
\tilde{J}_{0}\left(  k\right)  =\tilde{B}-ik\tilde{A}=\operatorname*{diag}%
\left[  \cos\theta_{1}+ik\sin\theta_{1},...,\cos\theta_{n_{M}}+ik\sin
\theta_{n_{M}},-I_{n_{D}},ikI_{n_{N}}\right]  , \label{JABtilde}%
\end{equation}%
\begin{equation}
\tilde{J}_{0}^{-1}\left(  k\right)  =\operatorname*{diag}\left[  \left(
\cos\theta_{1}+ik\sin\theta_{1}\right)  ^{-1},...,\left(  \cos\theta_{n_{M}%
}+ik\sin\theta_{n_{M}}\right)  ^{-1},-I_{n_{D}},\left(  ik\right)
^{-1}I_{n_{N}}\right]  , \label{Jtildeinverse}%
\end{equation}%
\begin{equation}
\tilde{S}_{0}\left(  k\right)  =-\tilde{J}_{0}\left(  -k\right)  \tilde{J}%
_{0}\left(  k\right)  ^{-1}=\operatorname*{diag}\left[  \frac{-\cos\theta
_{1}+ik\sin\theta_{1}}{\cos\theta_{1}+ik\sin\theta_{1}},...,\frac{-\cos
\theta_{n_{M}}+ik\sin\theta_{n_{M}}}{\cos\theta_{n_{M}}+ik\sin\theta_{n_{M}}%
},-I_{n_{D}},I_{n_{N}}\right]  . \label{Stilde}%
\end{equation}
Furthermore, $J_{0}\left(  k\right)  ^{-1}$ is related to the corresponding
$\tilde{J}_{0}^{-1}\left(  k\right)  $ by the relation (\cite{[9]})%
\begin{equation}
J_{0}\left(  k\right)  ^{-1}=T_{2}^{-1}MT_{1}^{-1}\tilde{J}_{0}^{-1}\left(
k\right)  M^{\dag}, \label{J0-1}%
\end{equation}
where $M$ is a unitary matrix and $T_{1},T_{2}$ are invertible. Similarly,
\begin{equation}
S_{0}(k)=M\,\tilde{S}_{0}(k)\,M^{\dagger}. \label{trs}%
\end{equation}

\subsection{Generalized Fourier
transforms.\label{Generalized Fourier transform}}

We now turn to the definition of the generalized Fourier transforms \cite{WederBook} and
\cite{WederStar}. Using the physical solution $\Psi$ we define%
\[
\left(  \mathbf{F}^{\pm}\psi\right)  \left(  k\right)  =\sqrt{\frac{1}{2\pi}%
}\int_{0}^{\infty}\left(  \Psi\left(  \mp k,x\right)  \right)  ^{\dagger}%
\psi\left(  x\right)  dx,
\]
for $\psi\in L^{2} \cap L^{1}.$ For any Borel set $O$ let $E(O)$ be the
spectral projector of $H$ for $O$. Then, (\cite{WederBook,WederStar})%

\begin{equation}
\label{isom}\left\Vert \mathbf{F}^{\pm}\psi\right\Vert _{L^{2}}= \left\Vert
E(\mathbb{R}^{+})\psi\right\Vert _{L^{2}}.
\end{equation}
Thus,  $\mathbf{F}^{\pm}$ extend to bounded operators on $L^{2}$ that we
also denote by $\mathbf{F}^{\pm}.$

The following spectral result for $H$ are proven in \cite{WederBook},
\cite{WederStar}.

\begin{proposition}
\label{PropespectroH}Suppose that the potential $V$ satisfies
(\ref{PotentialHermitian}) and (\ref{PotentialL1}). Then, the Hamiltonian $H$
has no positive bound states, and the negative spectrum of $H$ consists of
isolated bound states of multiplicity smaller or equal than $n$, that can
accumulate only at zero. Furthermore, $H$ has no singular continuous spectrum
and its absolutely continuous spectrum is given by $[0,\infty)$. The
generalized Fourier maps $\mathbf{F}^{\pm}$ are partially isometric with
initial subspace $\mathcal{H}_{\operatorname*{ac}}\left(  H\right)  $ and
final subspace $L^{2}$. Moreover, the adjoint operators are given by%
\[
\left(  \left(  \mathbf{F}^{\pm}\right)  ^{\dagger}\varphi\right)  \left(
x\right)  =\sqrt{\frac{1}{2\pi}}\int_{0}^{\infty}\Psi\left(  \mp k,x\right)
\varphi\left(  k\right)  dk,
\]
for $\varphi\in L^{2} \cap L^{1}.$ Furthermore,
\begin{equation}
\mathbf{F}^{\pm}H\left(  \mathbf{F}^{\pm}\right)  ^{\dagger}=\mathcal{M}
\label{spectralrepr}%
\end{equation}
where $\mathcal{M}$ is the operator of multiplication by $k^{2}.$ If, in
addition, $V\in L_{1}^{1},$ there is no bound state at $k=0$ and the number of
bounded states of $H$ is finite.
\end{proposition}

We observe that in particular (\ref{spectralrepr}) implies that%
\begin{equation}
\mathbf{F}^{\pm}e^{-itH}\left(  \mathbf{F}^{\pm}\right)  ^{\dagger
}=e^{-it\mathcal{M}}. \label{spectralgroup}%
\end{equation}
Note that by (\ref{isom}) $\left(  \mathbf{F}^{\pm}\right)  ^{\dagger
}\mathbf{F}^{\pm}$ is the orthogonal projector onto $\mathcal{H}%
_{\operatorname*{ac}}\left(  H\right)  .$ Since the singular continuous
spectrum is absent we get
\begin{equation}
\left(  \mathbf{F}^{\pm}\right)  ^{\dagger}\mathbf{F}^{\pm}%
=P_{\operatorname*{c}}, \label{projector}%
\end{equation}
with $P_{c}$ the projector onto the continuous subspace of $H.$ Therefore,
from (\ref{spectralgroup}) and (\ref{projector}) it follows
\begin{equation}
e^{-itH}P_{\operatorname*{c}}=\left(  \mathbf{F}^{\pm}\right)  ^{\dagger
}e^{-it\mathcal{M}}\mathbf{F}^{\pm}. \label{Spectralrepresentation}%
\end{equation}
Equation (\ref{Spectralrepresentation}) is the starting point for the proof of
our main results.

From (\ref{Spectralrepresentation}) (with the negative sign), for $\psi
\in\mathcal{S}$ ($\mathcal{S}$ denoting the Schwartz class) we have%
\[
e^{-itH}P_{\operatorname*{c}}\psi=\left(  2\pi\right)  ^{-1}\int_{0}^{\infty
}\Psi\left(  k,x\right)  e^{-itk^{2}}\left(  \int_{0}^{\infty}\left(
\Psi\left(  k,y\right)  \right)  ^{\dagger}\psi\left(  y\right)  dy\right)
dk.
\]
Using the definition (\ref{Physicalsolution}) of $\Psi\ ,$ and as by
(\ref{UnitaritySM}) $S\left(  k\right)  S^{\dagger}\left(  k\right)  =I,$
$S^{\dagger}\left(  -k\right)  =S\left(  k\right)  ,$ for $k\in\mathbb{R},$ we
get%
\begin{equation}
e^{-itH}P_{\operatorname*{c}}\psi=\left(  2\pi\right)  ^{-1}\int_{0}^{\infty
}\mathcal{T}\left(  x,y\right)  \psi\left(  y\right)  dy, \label{eq4}%
\end{equation}
where%
\begin{equation}
\mathcal{T}\left(  x,y\right)  =\int_{-\infty}^{\infty}e^{-itk^{2}}\left(
f\left(  -k,x\right)  \left(  f\left(  -k,y\right)  \right)  ^{\dagger
}+f\left(  k,x\right)  S\left(  k\right)  \left(  f\left(  -k,y\right)
\right)  ^{\dagger}\right)  dk. \label{eq4.a}%
\end{equation}

\subsection{The Fourier transform of $S(k)-S_{\infty} $}

\label{ScatteringFourier}

 We define below a set 
$\kappa_{j},$ $j=1,...,l$  of  $l$  distinct positive numbers related to the
bound-state energies $-\kappa_{j}^{2}, j=1,...,l, $ and a set $M_{j},$ $j=1,..,l$  of
constant $n\times n$ matrices related to the normalization of matrix-valued
bound-state eigenfunctions. These positive numbers and matrices where first introduced by  \cite{AgrMarch} in the case of Dirichlet boundary condition, and later by    \cite{Harmer} for  general boundary condition  (see also \cite{akwearxiv}, \cite{WederBook}).
As we mentioned in Proposition~\ref{Jostnozero}   the Jost matrix
$J\left(  k\right)  $ is analytic for $k\in\mathbb{C}^{+}$, continuous for
$k\in\overline{\mathbb{C}^{+}}$ and invertible for $k\in\mathbb{R}
\diagdown\{0\}.$ Further, it is proved in  \cite{AgrMarch} and  \cite{Harmer} (see also \cite{akwearxiv}, \cite{WederBook}), 
that the  determinant $\det[J(k)]$ is nonzero
in $\mathbb C^+$ except perhaps at a finite number of distinct 
$k$-values on the positive imaginary axis, that we denote by  $\kappa_j, j=1,...,l$.  In the case that $\det[J(k)]$ has no zeros in $\mathbb C^+$ we take $l=0.$ We  use $m_j$ to denote the multiplicity of
the zero of $\det[J(k)]$ at $k=i\kappa_j, j=1,...,l.$
Each  $m_j$ satisfies $1\leq m_j\le n, j=1....,l.$ The bound-state energies of the Schr\"odinger operator $H$ are given by $- \kappa_j^2,$ and they have multiplicity  $m_j, j=1,...,l.$
Moreover, we denote by  $\text{Ker}[J(i\kappa_j)^\dagger]$  the kernel
of the $n\times n$ constant matrix $J(i\kappa_j)^\dagger,$ and   we
designate by  $P_j$  the orthogonal projection
matrix   onto $\text{Ker}[J(i\kappa_j)^\dagger],$ for $j=1,\dots,l.$

Let us  define the constant $n\times n$ matrices
$A_j,$ $B_j,$ and $M_j$  as follows,
$$
A_j:=\int_0^\infty dx\,f(i\kappa_j,x)^\dagger\, f(i\kappa_j,x),\qquad j=1,\dots,l,
$$
$$
B_j:=(I-P_j)+P_jA_j\,P_j,\qquad j=1,\dots,l,
$$
$$
M_j:=B_j^{-1/2}P_j,\qquad j=1,\dots,l.
$$
The matrices $B_j, j=1,...,l$ are invertible. The  normalized matrix-valued bound-state eigenfunctions are given by,

$$
\Psi_j(x):=f(i\kappa_j,x)\,M_j,\qquad j=1,\dots,l.
$$

Let us denote by $F_{s}$  the Fourier transform
of $\left(  2\pi\right)  ^{-1/2}\left(  S\left(  k\right)  -S_{\infty}\right)
,$ that is
\begin{equation}
F_{s}\left(  y\right)  =\frac{1}{2\pi}\int_{-\infty}^{\infty}\left[  S\left(
k\right)  -S_{\infty}\right]  e^{iky}dk,\text{ \ }y\in\mathbb{R}\text{.}
\label{Fourierscattering}%
\end{equation}
where (\cite{[9]}),
\begin{equation}
S_{\infty}:=\lim_{|k|\rightarrow\infty}S\left(  k\right)  =\lim
_{|k|\rightarrow\infty}S_{0}\left(  k\right)  =MZ_{0}M^{\dag}, \label{sinft}%
\end{equation}
$Z_{0}:=\operatorname*{diag}[I_{n_{M}},-I_{n_{D}},I_{n_{N}}],$ (the numbers
$n_{M},$ $n_{D},$ $n_{N}$ are defined below (\ref{PSM13})) and $M$ is the
unitary matrix in (\ref{trans}). Here we denote by $I_{m}$ the $m\times m$
identity matrix. We define
\[
F\left(  y\right) : =F_{s}\left(  y\right)  +\sum_{j=1}^{l}M_{j}^{2}%
e^{-\kappa_{j}y},\text{ \ }y\in\mathbb{R}^{+}.
\]
In the case when $l=0$ we take $F=F_s.$
It is proved in \cite{AgrMarch}, \cite{akwearxiv} and \cite{WederBook} that,
\begin{equation}
F\in L^{1}\left(  0,\infty\right)  \cap L^{\infty}\left(  0,\infty\right)  .
\label{ConditionF}%
\end{equation}
Moreover, by \cite{AgrMarch}, \cite{akwearxiv} and \cite{WederBook}, the
function $K\left(  x,y\right)  $ satisfies the Marchenko equation
\begin{equation}
K\left(  x,y\right)  +F\left(  x+y\right)  +\int_{x}^{\infty}K\left(
x,t\right)  F\left(  t+y\right)  dt=0,\text{ }0\leq x<y. \label{Marchenko}%
\end{equation}

We now prove the following result concerning $F_{s}$.

\begin{theorem}
\label{FourierL1}\textbf{ }Suppose that the potential $V$ satisfies
(\ref{PotentialHermitian}) and (\ref{PotentialL11}). Then,%
\begin{equation}
F_{s}\in L^{1}\left(  \mathbb{R }\right)  . \label{FourierSM}%
\end{equation}

\end{theorem}

\begin{remark}
\textrm{   \rm We observe that this result is known in the case of the Dirichlet boundary
condition (see Theorem 5.6.2 on page 137 of \cite{AgrMarch}). In what
follows, we aim to extend (\ref{FourierSM}) to the case of the most general
boundary condition (\ref{SA1}).}
\end{remark}

In order to prove Theorem \ref{FourierL1}, we prepare some results. We begin
by proving the following adaptation of Lemma 5.6.2 on page 132 of
\cite{AgrMarch} to our settings:

\begin{proposition}
\label{PropMarchenko} Suppose that the $n \times n$ matrix $P_{0}$ satisfies
$P_{0} J^{\dagger}(0)=0.$ Then, the matrix $k^{-1}P_{0}J\left(  k\right)
^{\dag}$ fulfills
\begin{equation}
\label{wl1}k^{-1}P_{0}J\left(  k\right)  ^{\dag}=\mathcal{F}P_{0}G
+iP_{0}A^{\dagger},
\end{equation}
where $G(t) \in L^{1}\left(  \mathbb{R }\right)  $ and it is equal to zero for
$t >0.$
\end{proposition}

\begin{proof}
Integrating the Marchenko equation (\ref{Marchenko}) on $\left(
z,\infty\right)  ,$ with $z\geq x\geq0,$ we have%
\begin{equation}
\int_{z}^{\infty}K\left(  x,y\right)  dy+\int_{z+x}^{\infty}F\left(  y\right)
dy+\int_{x}^{\infty}K\left(  x,t\right)  \int_{z+t}^{\infty}F\left(  y\right)
dydt=0.\label{PSM0}%
\end{equation}
Evaluating in $x=0$ we get%
\begin{equation}
K_{1}\left(  z\right)  +\int_{z}^{\infty}F\left(  y\right)  dy+\int
_{0}^{\infty}K\left(  0,t\right)  \left(  \int_{z+t}^{\infty}F\left(
y\right)  dy\right)  dt=0,\label{PSM1}%
\end{equation}
where we denote%
\[
K_{1}\left(  z\right)  =\int_{z}^{\infty}K\left(  0,y\right)  dy.
\]
Moreover, differentiating (\ref{PSM0}) with respect to $x$ (this is possible
due to (\ref{ESTK}), (\ref{ESTK1}) and (\ref{ConditionF})) and taking $x=0$ we
have%
\begin{equation}
K_{2}\left(  z\right)  -F\left(  z\right)  -K\left(  0,0\right)  \int
_{z}^{\infty}F\left(  y\right)  dy+\int_{0}^{\infty}K_{x}\left(  0,t\right)
\int_{z+t}^{\infty}F\left(  y\right)  dydt=0,\label{PSM2}%
\end{equation}
with%
\[
K_{2}\left(  z\right)  =\int_{z}^{\infty}K_{x}\left(  0,y\right)  dy.
\]
Note that $K_{1}$ and $K_{2}$ are well-defined due to (\ref{ESTK}) and
(\ref{ESTK1}). Observe that $K_{1}^{\prime}\left(  t\right)  =-K\left(
0,t\right)  $ and $K_{2}^{\prime}\left(  t\right)  =-K_{x}\left(  0,t\right)
.$ Then, integrating by parts in the last integral in the left-hand side of
(\ref{PSM1}) and (\ref{PSM2}) we get,%
\begin{equation}
K_{1}\left(  z\right)  +\left(  I+K_{1}\left(  0\right)  \right)  \int
_{z}^{\infty}F\left(  y\right)  dy-\int_{0}^{\infty}K_{1}\left(  t\right)
F\left(  z+t\right)  dt=0,\label{new1}%
\end{equation}%
\begin{equation}
K_{2}\left(  z\right)  -F\left(  z\right)  -\left(  K\left(  0,0\right)
-K_{2}\left(  0\right)  \right)  \int_{z}^{\infty}F\left(  y\right)
dy-\int_{0}^{\infty}K_{2}\left(  t\right)  F\left(  z+t\right)
dt=0.\label{new2}%
\end{equation}
Multiplying from the left (\ref{new1}) by $B^{\dag}$ and (\ref{new2}) by
$A^{\dag}$ and considering the difference between the resulting equations we
get%
\begin{align}
&  B^{\dag}K_{1}\left(  z\right)  -A^{\dag}K_{2}\left(  z\right)  +\left(
B^{\dag}\left(  I+K_{1}\left(  0\right)  \right)  +A^{\dag}\left(  K\left(
0,0\right)  -K_{2}\left(  0\right)  \right)  \right)  \int_{z}^{\infty
}F\left(  y\right)  dy\nonumber\\
&  =-A^{\dag}F\left(  z\right)  +B^{\dag}\int_{0}^{\infty}K_{1}\left(
t\right)  F\left(  z+t\right)  dt-A^{\dag}\int_{0}^{\infty}K_{2}\left(
t\right)  F\left(  z+t\right)  dt.\label{PSM3}%
\end{align}
From the representation (\ref{Josttranspose}) for $J\left(  k\right)  ^{\dag}$
we see that%
\begin{equation}
J\left(  0\right)  ^{\dag}=B^{\dagger}\left(  I+K_{1}\left(  0\right)
\right)  +A^{\dagger}\left(  K\left(  0,0\right)  -K_{2}\left(  0\right)
\right)  .\label{PSM4}%
\end{equation}
By (\ref{Marchenko}) we get%
\begin{equation}
F\left(  z\right)  =-K\left(  0,z\right)  -\int_{0}^{\infty}K\left(
0,t\right)  F\left(  z+t\right)  dt.\label{PSM5}%
\end{equation}
Hence, from (\ref{PSM3}), via (\ref{PSM4}) and (\ref{PSM5}), we deduce
\begin{align*}
&  B^{\dag}K_{1}\left(  z\right)  -A^{\dag}K_{2}\left(  z\right)  +J\left(
0\right)  ^{\dag}\int_{z}^{\infty}F\left(  y\right)  dy\\
&  =A^{\dag}K\left(  0,z\right)  +\int_{0}^{\infty}\left(  A^{\dag}K\left(
0,t\right)  +B^{\dag}K_{1}\left(  t\right)  -A^{\dag}K_{2}\left(  t\right)
\right)  F\left(  z+t\right)  dt.
\end{align*}
Letting act $P_{0}$ from the left on the last equation and using that by
assumption $P_{0}J^{\dagger}(0)=0$ we get%
\begin{equation}
\mathcal{K}\left(  z\right)  =\int_{0}^{\infty}\mathcal{K}\left(  t\right)
F\left(  z+t\right)  dt+P_{0}\left(  A^{\dag}K\left(  0,z\right)  +\int
_{0}^{\infty}A^{\dag}K\left(  0,t\right)  F\left(  z+t\right)  dt\right)
,\label{PSM6}%
\end{equation}
where we denote%
\[
\mathcal{K}\left(  z\right)  =P_{0}\left(  B^{\dag}K_{1}\left(  z\right)
-A^{\dag}K_{2}\left(  z\right)  \right)  .
\]
From the estimates (\ref{ESTK}), (\ref{ESTK1}) for $K$ it follows that
$K\left(  0,z\right)  \in L^{1}\left(  0,\infty\right)  \cap L^{\infty}\left(
0,\infty\right)  $ and $K_{1},K_{2}\in L^{\infty}\left(  0,\infty\right)  .$
In particular,
\begin{equation}
\mathcal{K}\in L^{\infty}\left(  0,\infty\right)  .\label{eq14}%
\end{equation}
Let us prove that $\mathcal{K}\in L^{1}\left(  0,\infty\right)  .$ We proceed
similarly to the proof of Lemma 3.3.2 on page 72 of \cite{AgrMarch}. Since
$K\left(  0,z\right)  \in L^{1}\left(  0,\infty\right)  \cap L^{\infty}\left(
0,\infty\right)  $ and by (\ref{ConditionF}) $F\in L^{1}\left(  0,\infty
\right)  ,$ we see that the second term in the right-hand side of (\ref{PSM6})
belongs to $L^{1}\left(  0,\infty\right)  \cap L^{\infty}\left(
0,\infty\right)  .$ Moreover, by the density of the Schwartz class
$\mathcal{S}$ in $L^{1}(0,\infty),$ we can find $\tilde{F}\in\mathcal{S}$ such
that
\begin{equation}
\left\Vert F-\tilde{F}\right\Vert _{L^{1}(0,\infty)}<1.\label{approximation}%
\end{equation}
Then, we write (\ref{PSM6}) as
\begin{equation}
\mathcal{K}\left(  z\right)  +\int_{0}^{\infty}\mathcal{K}\left(  t\right)
F_{1}\left(  z+t\right)  dt=g\left(  z\right)  ,\label{eq15}%
\end{equation}
where $F_{1}=\tilde{F}-F$ and $g\in L^{1}(0,\infty)\cap L^{\infty}(0,\infty).$
Here we used that by (\ref{eq14}) and $\tilde{F}\in\mathcal{S}$,
\[
\left\Vert \int_{0}^{\infty}\mathcal{K}\left(  t\right)  \tilde{F}\left(
z+t\right)  dt\right\Vert _{L^{1}}+\left\Vert \int_{0}^{\infty}\mathcal{K}%
\left(  t\right)  \tilde{F}\left(  z+t\right)  dt\right\Vert _{L^{\infty}}\leq
C\left\Vert \mathcal{K}\right\Vert _{L^{\infty}}\int_{z}^{\infty}\left(
1+t\right)  \left\vert \tilde{F}\left(  t\right)  \right\vert dt\leq C.
\]
By (\ref{approximation}) $\left\Vert F_{1}\right\Vert _{L^{1}}<1.\ $Then, by
the method of successive approximations\ we see that there is a unique
solution $\mathcal{K}_{1}\in L^{1}\cap L^{\infty}$ to equation (\ref{eq15}).
Since $\mathcal{K}$ satisfies (\ref{eq14}) and (\ref{eq15}), we prove that
$\mathcal{K}\equiv\mathcal{K}_{1}$. Therefore, $\mathcal{K}\in L^{1}%
(0,\infty).$ Let
\begin{equation}
\mathcal{\hat{K}}\left(  k\right)  =\int_{0}^{\infty}e^{-ikz}\mathcal{K}%
\left(  z\right)  dz.\label{wl2}%
\end{equation}
Integrating by parts in the last integral we get%
\[
\mathcal{\hat{K}}\left(  k\right)  =\frac{1}{ik}P_{0}\left(  \left(  B^{\dag
}K_{1}\left(  0\right)  -A^{\dag}K_{2}\left(  0\right)  \right)  -\int
_{0}^{\infty}e^{-ikz}\left(  B^{\dag}K\left(  0,z\right)  -A^{\dag}%
K_{x}\left(  0,z\right)  \right)  dz\right)  .
\]
Then, by using (\ref{Josttranspose})\ and (\ref{PSM4}), since $P_{0}J\left(
0\right)  ^{\dag}=0$ we get%
\begin{equation}
\mathcal{\hat{K}}\left(  k\right)  =\frac{1}{ik}P_{0}\left(  J\left(
0\right)  ^{\dag}-J\left(  k\right)  ^{\dagger}\right)  +P_{0}A^{\dagger
}=-\frac{1}{ik}P_{0}J\left(  k\right)  ^{\dagger}+P_{0}A^{\dagger
}.\label{PSM7}%
\end{equation}
Denoting $G(t):=\mathcal{K}(-t),t<0$ and $G(t)=0,t>0,$ we obtain (\ref{wl1})
from (\ref{wl2}) and (\ref{PSM7}). This completes the proof.
\end{proof}

\bigskip In order to present our next result, we need the sharp small energy
behaviour of $J\left(  k\right)  ,$ obtained in \cite{[5]}. We have%
\begin{equation}
J\left(  k\right)  =\mathcal{G}P_{2}^{-1}%
\begin{bmatrix}
k\mathcal{A}_{1}+o\left(  k\right)  & k\mathcal{B}_{1}\mathcal{A}_{1}+o\left(
k\right) \\
k\mathcal{C}_{1}+o\left(  k\right)  & \mathcal{D}_{0}+o\left(  1\right)
\end{bmatrix}
P_{1}\mathcal{G}^{-1}, \label{jotak}%
\end{equation}
where the matrices $\mathcal{A}_{1},\mathcal{D}_{0},\mathcal{G}$,$P_{1}%
$,$P_{2}$ are invertible. Let us introduce the notation
\[
\alpha:=\mathcal{G}P_{2}^{-1}\text{ and }\beta:=P_{1}\mathcal{G}^{-1}.
\]
Then, it follows from (\ref{jotak}) that%
\[
J\left(  0\right)  =\alpha%
\begin{bmatrix}
0 & 0\\
0 & \mathcal{D}_{0}%
\end{bmatrix}
\beta
\]
and\qquad%
\begin{equation}
J\left(  0\right)  ^{\dagger}=\beta^{\dagger}%
\begin{bmatrix}
0 & 0\\
0 & \mathcal{D}_{0}^{\dagger}%
\end{bmatrix}
\alpha^{\dagger}. \label{jotacero}%
\end{equation}
We let%
\begin{equation}
P_{0}=\beta^{\dagger}%
\begin{bmatrix}
I & 0\\
0 & 0
\end{bmatrix}
\left(  \beta^{\dagger}\right)  ^{-1}. \label{P0}%
\end{equation}
Since $P_{0}J\left(  0\right)  ^{\dagger}=0,$ $P_{0}$ satisfies the
assumptions of Proposition \ref{PropMarchenko}. We observe that $P_{0}%
^{\dagger}$ is a projection onto the null space of $J\left(  0\right)  .$
Using this operator $P_{0}$ we define%
\begin{equation}
D\left(  k\right)  =\left(  I-P_{0}+\frac{1}{ik}P_{0}\right)  J\left(
k\right)  ^{\dagger}. \label{D}%
\end{equation}
Let us show that this matrix is non-singular. We prove the following:

\begin{proposition}
\label{Prop1}For all $k\in\mathbb{R}$ we have
\begin{equation}
\det D\left(  k\right)  \neq0. \label{PSM9}%
\end{equation}
\qquad
\end{proposition}

\begin{proof}
Since $P_{0}^{2}=P_{0},$ the equation $\left(  I-P_{0}+\frac{1}{ik}%
P_{0}\right)  \psi=0,$ for $k\neq0,$ implies that both $\left(  I-P_{0}%
\right)  \psi=0$ and $P_{0}\psi=0$ are satisfied. It follows that $\det\left(
I-P_{0}+\frac{1}{ik}P_{0}\right)  \neq0,$ for $k\neq0.$ Moreover, using
Proposition \ref{Jostnozero} we have
\begin{equation}
\det D\left(  k\right)  =\det\left(  I-P_{0}+\frac{1}{ik}P_{0}\right)  \det
J\left(  k\right)  ^{\dagger}\neq0,\label{PSM8}%
\end{equation}
for $k\neq0.$ Thus, we need to consider $D\left(  0\right)  .$ Using
(\ref{jotak}) and (\ref{P0}) we see that%
\[
\frac{1}{ik}P_{0}J\left(  k\right)  ^{\dagger}=\beta^{\dagger}%
\begin{bmatrix}
i^{-1}\mathcal{A}_{1}^{\dagger}+o\left(  1\right)   & i^{-1}\mathcal{C}%
_{1}^{\dagger}+o\left(  1\right)  \\
0 & 0
\end{bmatrix}
\alpha^{\dagger}.
\]
Then,
\begin{equation}
\lim_{k\rightarrow0}\frac{1}{ik}P_{0}J\left(  k\right)  ^{\dagger}%
=\beta^{\dagger}%
\begin{bmatrix}
i^{-1}\mathcal{A}_{1}^{\dagger} & i^{-1}\mathcal{C}_{1}^{\dagger}\\
0 & 0
\end{bmatrix}
\alpha^{\dagger}.\label{eq40}%
\end{equation}
Moreover, from (\ref{jotacero}) we calculate%
\begin{equation}
\lim_{k\rightarrow0}\left(  I-P_{0}\right)  J\left(  k\right)  ^{\dagger
}=\left(  I-P_{0}\right)  J\left(  0\right)  ^{\dagger}=\beta^{\dagger}%
\begin{bmatrix}
0 & 0\\
0 & \mathcal{D}_{0}^{\dagger}%
\end{bmatrix}
\alpha^{\dagger}.\label{eq41}%
\end{equation}
Therefore, by (\ref{eq40}), (\ref{eq41}) we get
\[
D\left(  0\right)  =\lim_{k\rightarrow0}\left(  I-P_{0}+\frac{1}{ik}%
P_{0}\right)  J\left(  k\right)  ^{\dagger}=\beta^{\dagger}%
\begin{bmatrix}
i^{-1}\mathcal{A}_{1}^{\dagger} & i^{-1}\mathcal{C}_{1}^{\dagger}\\
0 & \mathcal{D}_{0}^{\dagger}%
\end{bmatrix}
\alpha^{\dagger}.
\]
Hence,%
\[
\det D\left(  0\right)  =\det\left(  \alpha^{\dagger}\right)  \det\left(
\beta^{\dagger}\right)  \det\left(  i^{-1}\mathcal{A}_{1}^{\dagger}\right)
\det\mathcal{D}_{0}^{\dagger}.
\]
Since $\alpha,\beta,\mathcal{A}_{1},\mathcal{D}_{0}$ are invertible, we show
that%
\[
\det D\left(  0\right)  \neq0.
\]
This relation together with (\ref{PSM8}) imply (\ref{PSM9}).
\end{proof}

We prepare the following remark. We denote by $f \ast g$ the convolution of
$f$ and $g,$
\[
(f\ast g)(x):= \int_{\mathbb{R}}\, f(x-y)\, g(y)\, dy.
\]

\begin{remark}
\label{remprod}\textrm{Suppose that $f,g\in L^{1}(\mathbb{R}).$ Then, since
for $f,g\in L^{1}(\mathbb{R}),$ $\left(  \mathcal{F}f\right)  \left(
\mathcal{F}g\right)  =\mathcal{F}\left(  f\ast g\right)  ,$ and $f\ast g\in
L^{1}(\mathbb{R}),$ we have that $\left(  \mathcal{F}f\right)  \left(
\mathcal{F}g\right)  \in\mathcal{F}\left(  L^{1}(\mathbb{R})\right)  .$ That
is to say, $\mathcal{F}\left(  L^{1}(\mathbb{R})\right)  $ is closed under
products.}
\end{remark}

We also need the Wiener-L\'{e}vy theorem (see 6.1.8 on page 262 of
\cite{Trigub}):

\begin{proposition}
[Wiener-L\'{e}vy theorem]Suppose that $f\in\mathcal{F}\left(  L^{1}\left(
\mathbb{R}\right)  \right)  .$ Let $F$ be an analytic function on an open set
of $\mathbb{C}$ which contains the range of $f.$ Then $F\circ f$
$\in\mathcal{F}\left(  L^{1}\left(  \mathbb{R}\right)  \right)  .$
\end{proposition}

In fact, what we actually use is the following corollary of the
Wiener-L\'{e}vy theorem:

\begin{corollary}
\label{Wiener-Levy}Suppose that $f\in\mathcal{F}\left(  L^{1}\left(
\mathbb{R}\right)  \right)  .$ Given $l\in\mathbb{C}$, if $f\left(  x\right)
\neq l,$ for all $x\in\mathbb{R}$, then $\displaystyle\frac{f}{f-l}%
\in\mathcal{F}\left(  L^{1}\left(  \mathbb{R}\right)  \right)  .$
\end{corollary}

\begin{proof}
In the Wiener-L\'{e}vy theorem take $F(z)=(z-l)^{-1}.$ Then, $(f-l)^{-1}%
\in\mathcal{F}\left(  L^{1}\left(  \mathbb{R}\right)  \right)  $ and by Remark
\ref{remprod} $f(f-l)^{-1}\in\mathcal{F}\left(  L^{1}\left(  \mathbb{R}%
\right)  \right)  .$
\end{proof}

Finally, we present a local Wiener theorem (see Theorem 229 on page 290 of
\cite{Hardy}):

\begin{proposition}
\label{Hardy}Suppose that $f,g\in\mathcal{F}\left(  L^{1}\left(
\mathbb{R}\right)  \right)  ,$ that $g\left(  x\right)  \neq0$, for all
$x\in\mathbb{R}$, and that $f\left(  x\right)  =0,$ for $\left\vert
x\right\vert >\lambda>0.$ Then, $\displaystyle \frac{f}{g}\in\mathcal{F}%
\left(  L^{1}\left(  \mathbb{R}\right)  \right)  .$
\end{proposition}

We have now all the necessary ingredients to prove (\ref{FourierSM}).

\begin{proof}
[Proof of Theorem \ref{FourierL1}]We depart from the definition
(\ref{Scatteringmatrix}) of the scattering matrix $S\left(  k\right)  .$ Let
$\chi\in C^{\infty}\left(  \mathbb{R}\right)  ,$ $0\leq\chi\leq1,$ be such
that $\chi\left(  k\right)  =1,$ for $\left\vert k\right\vert \leq1$ and
$\chi\left(  k\right)  =0,$ for $\left\vert k\right\vert \geq2$. For $a>0,$ we
set $\chi_{a}\left(  k\right)  :=\chi\left(  \frac{k}{a}\right)  ,$
$k\in\mathbb{R}$. We consider $S\left(  -k\right)  .$ Using
(\ref{Scatteringmatrix1}) we decompose
\begin{equation}
S\left(  -k\right)  -S_{\infty}^{\dagger}=S_{1}\left(  k\right)  +S_{2}\left(
k\right)  ,\label{PSM17}%
\end{equation}
where%
\begin{align}
S_{1}\left(  k\right)   &  :=-\chi_{a}\left(  k\right)  \left(  \left(
J\left(  k\right)  ^{\dagger}\right)  ^{-1}\chi_{2a}\left(  k\right)  J\left(
-k\right)  ^{\dagger}+S_{\infty}^{\dagger}\right)  ,\label{S1}\\
S_{2}\left(  k\right)   &  :=-\left(  1-\chi_{a}\left(  k\right)  \right)
\left(  \left(  J\left(  k\right)  ^{\dagger}\right)  ^{-1}J\left(  -k\right)
^{\dagger}+S_{\infty}^{\dagger}\right)  .\label{S2}%
\end{align}
Using (\ref{D}) we write%
\begin{align*}
S_{1}\left(  k\right)   &  =-\chi_{a}\left(  k\right)  \left(  \left(
D\left(  k\right)  \right)  ^{-1}\left(  I-P_{0}+\frac{1}{ik}P_{0}\right)
\left(  I-P_{0}-\frac{1}{ik}P_{0}\right)  ^{-1}\chi_{2a}\left(  k\right)
D\left(  -k\right)  +S_{\infty}^{\dagger}\right)  \\
&  =-\chi_{a}\left(  k\right)  \left(  \left(  D\left(  k\right)  \right)
^{-1}\left(  I-2P_{0}\right)  \chi_{2a}\left(  k\right)  D\left(  -k\right)
+S_{\infty}^{\dagger}\right)  .
\end{align*}
Observe that
\begin{equation}
\mathcal{F}^{-1}\chi_{a},\mathcal{F}^{-1}\left(  k\chi_{a}\left(  k\right)
\right)  \in L^{1}\cap L^{2},\text{ }a>0.\label{PSM12}%
\end{equation}
Moreover by Remark \ref{remprod} we see that $\mathcal{F}\left(  L^{1}\left(
\mathbb{R}\right)  \right)  $ is closed by products. Therefore, it follows
from (\ref{ESTK}), (\ref{ESTK1}) and (\ref{Josttranspose}) that
\begin{equation}
\chi_{a}\left(  k\right)  J\left(  k\right)  ^{\dagger}\in\mathcal{F}\left(
L^{1}\left(  \mathbb{R}\right)  \right)  ,\text{ }a>0.\label{PSM11}%
\end{equation}
Hence, from Proposition \ref{PropMarchenko} we show that all the elements of
the matrix $\chi_{2a}\left(  k\right)  D\left(  -k\right)  $ belong to
$\mathcal{F}\left(  L^{1}\left(  \mathbb{R}\right)  \right)  .$ On the other
hand, all the entries of the matrix $\chi_{a}\left(  k\right)  \left(
D\left(  k\right)  \right)  ^{-1}$ can be represented as
\[
\frac{\chi_{a}\left(  k\right)  L\left(  k\right)  }{\det D\left(  k\right)
},
\]
with $L\in\mathcal{F}\left(  L^{1}\left(  \mathbb{R}\right)  \right)  .$ Due
to the cut-off function $\chi_{a},$ we express the last relation as
\[
\frac{\chi_{a}\left(  k\right)  L\left(  k\right)  }{\det D\left(  k\right)
}=\frac{\chi_{a}\left(  k\right)  L\left(  k\right)  }{\chi_{3a}\left(
k\right)  \det D\left(  k\right)  +\left(  1-\chi_{3a}\left(  k\right)
\right)  G\left(  k\right)  },
\]
with any $G\in\mathcal{F}\left(  L^{1}\left(  \mathbb{R}\right)  \right)  ,$
$G\left(  k\right)  \neq0,$ for all $k\in\mathbb{R}$ (for example, taking any
non-vanishing function from the Schwartz class). By Proposition
\ref{PropMarchenko} and (\ref{PSM11}), $\chi_{3a}\left(  k\right)  \det
D\left(  k\right)  \in\mathcal{F}\left(  L^{1}\left(  \mathbb{R}\right)
\right)  .$ By (\ref{PSM9}), $\det D\left(  k\right)  \neq0,$ for all
$k\in\mathbb{R}$. Then, on the support of $\chi_{3a},$ the function $\chi
_{3a}\left(  k\right)  \det D\left(  k\right)  $ has a definite sign. We take
$G$ with a definite sign such that $\chi_{3a}\left(  k\right)  \det D\left(
k\right)  +\left(  1-\chi_{3a}\left(  k\right)  \right)  G\left(  k\right)
\neq0,$ for all $k\in\mathbb{R}$. Then, each element of the matrix $\chi
_{a}\left(  k\right)  \left(  D\left(  k\right)  \right)  ^{-1}$ is of the
form $\chi_{a}\left(  k\right)  \frac{f}{g},$ with $f,g\in\mathcal{F}\left(
L^{1}\left(  \mathbb{R}\right)  \right)  ,$ such that $g\left(  k\right)
\neq0$, for all $k\in\mathbb{R}$, and $f\left(  k\right)  =0,$ for all
$\left\vert k\right\vert \geq2a.$ By Proposition \ref{Hardy}, $\frac{f}{g}%
\in\mathcal{F}\left(  L^{1}\left(  \mathbb{R}\right)  \right)  .$ Therefore,
by (\ref{PSM12}) we conclude that all the elements of $\chi_{a}\left(
k\right)  \left(  D\left(  k\right)  \right)  ^{-1}$ belong to $\mathcal{F}%
\left(  L^{1}\left(  \mathbb{R}\right)  \right)  .$ Since also by
(\ref{PSM12}), the entries of $\chi_{a}\left(  k\right)  S_{\infty}^{\dagger}$
are functions in $\mathcal{F}\left(  L^{1}\left(  \mathbb{R}\right)  \right)
,$ we conclude that $S_{1}\left(  k\right)  $ is a matrix which elements can
be represented as Fourier transform of functions in $L^{1}(\mathbb{R}).$ Next,
we consider $S_{2}\left(  k\right)  .$ We put $a>2$ in (\ref{S2}). Using
(\ref{Josttranspose}) we have
\[
J\left(  k\right)  ^{\dagger}=\left(  B^{\dagger}+ikA^{\dagger}\right)
\left(  I+\left(  B^{\dagger}+ikA^{\dagger}\right)  ^{-1}\left(  A^{\dagger
}K\left(  0,0\right)  +G_{1}\left(  k\right)  \right)  \right)  ,
\]
where the elements of the matrix $G_{1}$ belong to $\mathcal{F}\left(
L^{1}\left(  \mathbb{R}\right)  \right)  .$ By (\ref{J0-1})%
\begin{equation}
\left(  B^{\dagger}+ikA^{\dagger}\right)  ^{-1}=M\left(  \tilde{J}_{0}\left(
k\right)  ^{\dagger}\right)  ^{-1}\left(  T_{1}^{\dagger}\right)  ^{-1}%
M^{\dag}\left(  T_{2}^{\dagger}\right)  ^{-1},\label{PSM15}%
\end{equation}
where $\tilde{J}_{0}^{-1}\left(  k\right)  $ is given by the diagonal matrix
(\ref{Jtildeinverse}). Then, as by Proposition \ref{Jostnozero} $J\left(
k\right)  ^{\dagger}$ is invertible for $k\neq0,$ we see that
\begin{equation}
\det\left(  I+\left(  B^{\dagger}+ikA^{\dagger}\right)  ^{-1}\left(
A^{\dagger}K\left(  0,0\right)  +G_{1}\left(  k\right)  \right)  \right)
\neq0,\text{ }k\in\mathbb{R\setminus}\{0\}.\label{nozero}%
\end{equation}
Therefore, using (\ref{w00}) and (\ref{PSM15}) we decompose
\begin{align}
&  -\left(  1-\chi_{a}\left(  k\right)  \right)  \left(  J\left(  k\right)
^{\dagger}\right)  ^{-1}J\left(  -k\right)  ^{\dagger}\nonumber\\
&  =\left(  1-\chi_{a}\left(  k\right)  \right)  \left(  I+MJ_{\chi}\left(
k\right)  \left(  T_{1}^{\dagger}\right)  ^{-1}M^{\dag}\left(  T_{2}^{\dagger
}\right)  ^{-1}\left(  A^{\dagger}K\left(  0,0\right)  +G_{1}\left(  k\right)
\right)  \right)  ^{-1}\nonumber\\
&  \times S_{0}\left(  k\right)  ^{\dagger}\left(  I+MJ_{\chi}\left(
-k\right)  \left(  T_{1}^{\dagger}\right)  ^{-1}M^{\dag}\left(  T_{2}%
^{\dagger}\right)  ^{-1}\left(  A^{\dagger}K\left(  0,0\right)  +G_{1}\left(
-k\right)  \right)  \right),\label{PSM14}%
\end{align}
with%
\[
J_{\chi}\left(  k\right)  =\operatorname*{diag}\left[  \left(  \cos\theta
_{1}-ik\sin\theta_{1}\right)  ^{-1},...,\left(  \cos\theta_{n_{M}}%
-ik\sin\theta_{n_{M}}\right)  ^{-1},-I_{n_{D}},-\left(  ik\right)
^{-1}\left(  1-\chi\left(  k\right)  \right)  I_{n_{N}}\right]  ,
\]
where we can introduce the functions $1-\chi\left(  k\right)  $ in the entries
of $\left(  \tilde{J}_{0}\left(  k\right)  ^{\dagger}\right)  ^{-1}$
corresponding to the Neumann boundary conditions without modifying the
equality thanks to the cut-off function $1-\chi_{a}\left(  k\right)  $ (we put
$a>2$). We now observe that
\begin{equation}
\mathcal{F}^{-1}\left(  \cos\theta_{j}-ik\sin\theta_{j}\right)  ^{-1}\in
L^{1}(\mathbb{R})\text{, }j=1,...,n_{M},\label{eq16}%
\end{equation}
and%
\begin{equation}
\mathcal{F}^{-1}\left(  \frac{1-\chi\left(  k\right)  }{k}\right)  \in
L^{1}(\mathbb{R}),\label{eq17}%
\end{equation}
where the first relation is due to the Jordan lemma and contour integration
and the second one follows by integration by parts. Then, we show that the
entries of $MJ_{\chi}\left(  k\right)  \left(  T_{1}^{\dagger}\right)
^{-1}M^{\dag}\left(  T_{2}^{\dagger}\right)  ^{-1}G_{1}\left(  k\right)  $
belong to $\mathcal{F}\left(  L^{1}\left(  \mathbb{R}\right)  \right)  $.
Using (\ref{trans}) we calculate%
\[
MJ_{\chi}\left(  k\right)  \left(  T_{1}^{\dagger}\right)  ^{-1}M^{\dag
}\left(  T_{2}^{\dagger}\right)  ^{-1}A^{\dagger}K\left(  0,0\right)
=MJ_{\chi}\left(  k\right)  \tilde{A}M^{\dagger}K\left(  0,0\right)  .
\]
Since%
\[
J_{\chi}\left(  k\right)  \tilde{A}=\operatorname*{diag}\left[  \left(
\cos\theta_{1}-ik\sin\theta_{1}\right)  ^{-1},...,\left(  \cos\theta_{n_{M}%
}-ik\sin\theta_{n_{M}}\right)  ^{-1},0,-\left(  ik\right)  ^{-1}\left(
1-\chi\left(  k\right)  \right)  I_{n_{N}}\right]  ,
\]
it also follows from (\ref{eq16}) and (\ref{eq17}) that the elements of
$MJ_{\chi}\left(  k\right)  \left(  T_{1}^{\dagger}\right)  ^{-1}M^{\dag
}\left(  T_{2}^{\dagger}\right)  ^{-1}A^{\dagger}K\left(  0,0\right)  $ belong
to $\mathcal{F}\left(  L^{1}\left(  \mathbb{R}\right)  \right)  .$ Thus, we
can write
\[
I+MJ_{\chi}\left(  k\right)  \left(  T_{1}^{\dagger}\right)  ^{-1}M^{\dag
}\left(  T_{2}^{\dagger}\right)  ^{-1}\left(  A^{\dagger}K\left(  0,0\right)
+G_{1}\left(  k\right)  \right)  =I+G_{2}\left(  k\right)  ,
\]
where the elements of $G_{2}\ $are in $\mathcal{F}\left(  L^{1}\left(
\mathbb{R}\right)  \right)  .$ Then, from (\ref{PSM14}) we get%
\begin{equation}
-\left(  1-\chi_{a}\left(  k\right)  \right)  \left(  J\left(  k\right)
^{\dagger}\right)  ^{-1}J\left(  -k\right)  ^{\dagger}=\left(  1-\chi
_{a}\left(  k\right)  \right)  \left(  I+G_{2}\left(  k\right)  \right)
^{-1}S_{0}\left(  k\right)  ^{\dagger}\left(  I+G_{2}\left(  -k\right)
\right)  .\label{PSM19}%
\end{equation}
Since $\mathcal{F}\left(  L^{1}\left(  \mathbb{R}\right)  \right)  $ is closed
by products, we write
\begin{equation}
\det\left(  I+G_{2}\left(  k\right)  \right)  =1+g_{2}\left(  k\right)
,\label{eq18}%
\end{equation}
with $g_{2}\in\mathcal{F}\left(  L^{1}\left(  \mathbb{R}\right)  \right)  .$
We observe that on the support of $1-\chi_{a}\left(  k\right)  $ we can
represent%
\begin{equation}
1+g_{2}\left(  k\right)  =1+\left(  1-\chi_{a/2}\left(  k\right)  \right)
g_{2}\left(  k\right)  .\label{eq19}%
\end{equation}
By Riemann--Lebesgue lemma $g_{2}\left(  k\right)  \rightarrow0,$ as
$\left\vert k\right\vert \rightarrow\infty.$ Then, we can take $a>2$
sufficiently large in a way that $1+\left(  1-\chi_{a/2}\left(  k\right)
\right)  g_{2}\left(  k\right)  \neq0,$ for all $k\in\mathbb{R}$. By
(\ref{PSM12}) $\left(  1-\chi_{a/2}\right)  g_{2}\in\mathcal{F}\left(
L^{1}\left(  \mathbb{R}\right)  \right)  .$ Then, by Corollary
\ref{Wiener-Levy} we show that
\begin{equation}
\left(  1+\left(  1-\chi_{a/2}\left(  k\right)  \right)  g_{2}\left(
k\right)  \right)  ^{-1}=1+g_{3}\left(  k\right)  ,\label{eq20}%
\end{equation}
where $g_{3}\in\mathcal{F}\left(  L^{1}\left(  \mathbb{R}\right)  \right)  .$
Hence, from (\ref{eq18}), (\ref{eq19}) and (\ref{eq20}) it follows that%
\[
\frac{1}{\det\left(  I+G_{2}\left(  k\right)  \right)  }=1+g_{3}\left(
k\right)  ,
\]
for all $\left\vert k\right\vert \geq a.$ Using the last expression in
(\ref{PSM19}) we get
\begin{equation}
-\left(  1-\chi_{a}\left(  k\right)  \right)  \left(  J\left(  k\right)
^{\dagger}\right)  ^{-1}J\left(  -k\right)  ^{\dagger}=\left(  1-\chi
_{a}\left(  k\right)  \right)  \left(  I+G_{3}\left(  k\right)  \right)
S_{0}\left(  k\right)  ^{\dagger}\left(  I+G_{2}\left(  -k\right)  \right)
,\label{PSM16}%
\end{equation}
where the elements of $G_{3}\ $are in $\mathcal{F}\left(  L^{1}\left(
\mathbb{R}\right)  \right)  .$ From (\ref{Stilde}), (\ref{trs}), via
(\ref{eq16}) and (\ref{eq17}) we show that all the elements of the matrix
$\left(  1-\chi_{a}\left(  k\right)  \right)  \left(  S_{0}\left(  k\right)
^{\dagger}-S_{\infty}^{\dagger}\right)  $ belong to $\mathcal{F}\left(
L^{1}\left(  \mathbb{R}\right)  \right)  .$ Therefore, from (\ref{PSM16}) we
conclude that the entries of $S_{2}\left(  k\right)  $ are in $\mathcal{F}%
\left(  L^{1}\left(  \mathbb{R}\right)  \right)  .$ Finally, from
(\ref{PSM17}) we obtain (\ref{FourierSM}).
\end{proof}

\section{The $L^{p}-L^{p^{\prime}}$ estimate for the matrix Schr\"{o}dinger
equation.\label{Lp-Lq}}

\bigskip\ This section is devoted to the proof of the $L^{p}-L^{p^{\prime}}$
and Strichartz estimates for the matrix Schr\"{o}dinger equation. We begin by
proving the following $L^{1}-L^{\infty}$ estimate.

\begin{proposition}
\label{PropL1}Suppose that the potential $V$ satisfies
(\ref{PotentialHermitian}) and (\ref{PotentialL11}). Then, the estimates
\begin{equation}
\left\Vert e^{-itH}P_{\operatorname*{c}}\right\Vert _{\mathcal{B}\left(
L^{1},L^{\infty}\right)  }\leq\frac{C}{\sqrt{\left\vert t\right\vert }},
\label{L1-Linf}%
\end{equation}
and%
\begin{equation}
\left\Vert e^{-itH}P_{\operatorname*{c}}\right\Vert _{\mathcal{B}\left(
W_{1,1},W_{1,\infty}\right)  }\leq\frac{C}{\sqrt{\left\vert t\right\vert }},
\label{L1-LinfSobolev}%
\end{equation}
are true for all $t\in\mathbb{R\setminus\{}0\}.$
\end{proposition}

\begin{proof}
We depart from the spectral representation (\ref{eq4},\ref{eq4.a}) for
$e^{-itH}P_{\operatorname*{c}}.$ We decompose $\mathcal{T}\left(  x,y\right)
$ as follows%
\begin{equation}
\mathcal{T}\left(  x,y\right)  =\sum_{j=0}^{5}\mathcal T_{j}(x,y), \label{T}%
\end{equation}
with%
\[
\mathcal T_{0}:=\int_{-\infty}^{\infty}e^{-itk^{2}}e^{-ik\left(  x-y\right)
}dk+S_{\infty}\int_{-\infty}^{\infty}e^{-itk^{2}}e^{ik\left(  x+y\right)
}dk,
\]%
\begin{align*}
\mathcal T_{1}  &  :=\int_{-\infty}^{\infty}e^{-itk^{2}}\left(  e^{-ikx}d\left(
-k,y\right)  ^{\dagger}+e^{-iky}d\left(  k,x\right)  \right)  dk\\
&  +\int_{-\infty}^{\infty}e^{-itk^{2}}\left(  e^{-ikx}S_{\infty}d\left(
k,y\right)  ^{\dagger}+e^{-iky}d\left(  -k,x\right)  S_{\infty}\right)  dk,
\end{align*}%
\[
\mathcal T_{2}:=\int_{-\infty}^{\infty}e^{-itk^{2}}\left(  d\left(  -k,x\right)
d\left(  -k,y\right)  ^{\dagger}+d\left(  k,x\right)  S_{\infty}d\left(
-k,y\right)  ^{\dagger}\right)  dk,
\]%
\[
\mathcal T_{3}:=\int_{-\infty}^{\infty}e^{-itk^{2}}e^{ik\left(  x+y\right)  }T\left(
k\right)  dk,
\]%
\[
\mathcal T_{4}:=\int_{-\infty}^{\infty}e^{-itk^{2}}\left(  e^{ikx}T\left(  k\right)
d\left(  -k,y\right)  ^{\dagger}+e^{iky}d\left(  k,x\right)  T\left(
k\right)  \right)  dk,
\]%
\[
\mathcal T_{5}:=\int_{-\infty}^{\infty}e^{-itk^{2}}d\left(  k,x\right)  T\left(
k\right)  d\left(  -k,y\right)  ^{\dagger}dk,
\]
where we denote,
\begin{equation}
d(k,x):=f(k,x)-e^{ikx}, \label{dw1}%
\end{equation}
and
\begin{equation}
T(k):=S(k)-S_{\infty}. \label{tw1}%
\end{equation}
Recall that,
\begin{equation}
\frac{1}{2\pi}\int_{-\infty}^{\infty}e^{-itk^{2}}e^{-ikz}dk=\frac
{e^{iz^{2}/4t}}{\sqrt{4\pi it}}, \label{eq5}%
\end{equation}
with the Fourier transform understood in the sense of distributions. Using
(\ref{eq5}) we have that%
\begin{equation}
\mathcal T_{0}=\sqrt{\frac{\pi}{it}}\left(  e^{i\left(  x-y\right)  ^{2}/4t}%
+e^{i\left(  x+y\right)  ^{2}/4t}S_{\infty}\right)  . \label{wi0}%
\end{equation}
We observe that $\mathcal T_{0}+\mathcal T_{3}$ corresponds to the free evolution $V\equiv0.$
Moreover, if in the diagonal representation (\ref{A,Btilde}) there are only
Dirichlet and Neumann boundary conditions, and $V\equiv0$ it follows from
(\ref{w00}) that $T\equiv0$ and then, $\mathcal T_{3}\equiv0$ in this case. By
(\ref{ESTK}), for fixed $x\geq0,$ $K\left(  x,\cdot\right)  \in L^{2}\left(
x,\infty\right)  .$ Then, using (\ref{jostK}) we get
\[
d\left(  k,x\right)  =\int_{-\infty}^{\infty}e^{ikz}K\left(  x,z\right)  dz,
\]
(recall that $K\left(  x,z\right)  =0,$ for $z<x).$ Hence, by the convolution
theorem for the Fourier transform and (\ref{eq5}) we obtain%
\begin{align}
\mathcal T_{1}  &  =\frac{\left(  2\pi\right)  }{\sqrt{4\pi it}}\left(  \int
_{y}^{\infty}e^{i\left(  x-z\right)  ^{2}/4t}K\left(  y,z\right)  ^{\dagger
}dz+\int_{x}^{\infty}e^{i\left(  y-z\right)  ^{2}/4t}K\left(  x,z\right)
dz\right. \nonumber\\
&  \left.  +\int_{x}^{\infty}e^{i\left(  y+z\right)  ^{2}/4t}K\left(
x,z\right)  S_{\infty}dz+\int_{y}^{\infty}e^{i\left(  x+z\right)  ^{2}%
/4t}S_{\infty}K\left(  y,z\right)  ^{\dagger}dz\right)  . \label{eq6}%
\end{align}
Hence, from (\ref{ESTK}) it follows that%
\begin{equation}
\left\vert \mathcal T_{1}\right\vert \leq C\frac{1}{\sqrt{|t|}}. \label{wi1}%
\end{equation}
Moreover, differentiating (\ref{eq6}) with respect to $x$, noting that
$\partial_{x}e^{i\left(  x\pm z\right)  ^{2}/4t}=\pm\partial_{z}e^{i\left(
x\pm z\right)  ^{2}/4t}$ and using (\ref{ESTK}), (\ref{ESTK1}), we prove that%
\begin{equation}
\left\vert \partial_{x}\mathcal T_{1}\right\vert \leq C\frac{1}{\sqrt{|t|}}.
\label{wi2}%
\end{equation}
Next, we consider $\mathcal T_{2}.$ By Parseval's identity and the convolution theorem,
via (\ref{eq5}) we get%
\begin{align*}
\mathcal T_{2}  &  =\frac{2\pi}{\sqrt{4\pi it}}\left(  \int_{-\infty}^{\infty}%
e^{iz_{1}^{2}/4t}K\left(  x,-z_{2}\right)  K\left(  y,z_{1}-z_{2}\right)
^{\dagger}dz_{2}dz_{1}\right. \\
&  +\left.  \int_{-\infty}^{\infty}e^{iz_{1}^{2}/4t}K\left(  x,z_{2}\right)
S_{\infty}K\left(  y,z_{1}-z_{2}\right)  ^{\dagger}dz_{2}dz_{1}\right)  .
\end{align*}
Then, by (\ref{ESTK}) and (\ref{ESTK1}) we get%
\begin{equation}
\left\vert \partial_{x}^{j}\mathcal T_{2}\right\vert \leq C\frac{1}{\sqrt{|t|}},\text{
\ }j=0,1. \label{wi3}%
\end{equation}
Next, we consider $\mathcal T_{3}.$ By the convolution theorem we have
\begin{equation}
\mathcal T_{3}=\frac{2\pi}{\sqrt{4\pi it}}\int_{-\infty}^{\infty}e^{i\left(
x+y-z\right)  ^{2}/4t}F_{s}\left(  z\right)  dz, \label{wi4}%
\end{equation}
where $F_{s}$ is given by (\ref{Fourierscattering}). Further, by
(\ref{FourierSM}) we prove that%
\begin{equation}
\left\vert \mathcal T_{3}\right\vert \leq C\frac{1}{\sqrt{|t|}}. \label{wi5}%
\end{equation}
Denote by $\mathcal{I}_{3}$ the integral operator from $L^{1}(\mathbb{R}^{+})$
into $L^{\infty}(\mathbb{R}^{+})$ with integral kernel $\mathcal T_{3}(x,y),$
\begin{equation}
(\mathcal{I}_{3}\phi)(x):=\int_{0}^{\infty}\,\mathcal T_{3}(x,y)\,\phi(y)\,dy.
\label{I3}%
\end{equation}
By (\ref{wi5})
\begin{equation}
\Vert\mathcal{I}_{3}\phi\Vert_{L^{\infty}}\leq C\frac{1}{\sqrt{|t|}}\Vert
\phi\Vert_{L^{1}}. \label{wi6}%
\end{equation}
Moreover, derivating the right-hand side of (\ref{I3}), using (\ref{wi4}),
noting that $\partial_{x}e^{i\left(  x+y-z\right)  ^{2}/4t}=\partial
_{y}e^{i\left(  x+y-z\right)  ^{2}/4t},$ integrating by parts in $y$ and using
(\ref{FourierSM}), by Sobolev embedding theorem we show that%
\begin{equation}
\left\Vert \partial_{x}\mathcal{I}_{3}\phi\right\Vert _{L^{\infty}}\leq
C\frac{1}{\sqrt{|t|}}\left(  \left\vert \phi\left(  0\right)  \right\vert
+\left\Vert \phi^{\prime}\right\Vert _{L^{1}}\right)  \leq C\frac{1}%
{\sqrt{|t|}}\left\Vert \phi\right\Vert _{W_{1,1}}. \label{wi7}%
\end{equation}
Further, by (\ref{wi6}) and (\ref{wi7}),
\begin{equation}
\left\Vert \mathcal{I}_{3}\right\Vert _{\mathcal{B}(W_{1,1},W_{1,\infty})}\leq
C\frac{1}{\sqrt{|t|}}. \label{wi8}%
\end{equation}
Next, we turn to $\mathcal T_{4}.$ By the convolution theorem%
\begin{align}
\mathcal T_{4}  &  =\frac{2\pi}{\sqrt{4\pi it}}\int_{-\infty}^{\infty}e^{i\left(
x-z\right)  ^{2}/4t}\int_{-\infty}^{\infty}F_{s}\left(  z-z_{1}\right)
K\left(  y,-z_{1}\right)  ^{\dagger}dz_{1}dz\nonumber\\
&  +\frac{2\pi}{\sqrt{4\pi it}}\int_{-\infty}^{\infty}e^{i\left(  y-z\right)
^{2}/4t}\int_{-\infty}^{\infty}K\left(  x,-z_{1}\right)  F_{s}\left(
z-z_{1}\right)  dz_{1}dz. \label{I4}%
\end{align}
Then, using (\ref{ESTK}) and (\ref{FourierSM}) we prove that%
\begin{equation}
\left\vert \mathcal T_{4}\right\vert \leq C\frac{1}{\sqrt{|t|}}. \label{wi9}%
\end{equation}
Moreover, differentiating (\ref{I4}) with respect to $x,$ noting that the
following identities are true $\partial_{x}e^{i\left(  x-z\right)  ^{2}%
/4t}=-\partial_{z}e^{i\left(  x-z\right)  ^{2}/4t}$ and $\partial_{z}%
F_{s}\left(  z-z_{1}\right)  =-\partial_{z_{1}}F_{s}\left(  z-z_{1}\right)  ,$
integrating by parts first in $z$ and then in $z_{1}$ in order to make the
derivative act on $K,$ and using (\ref{ESTK}), (\ref{ESTK1}) and
(\ref{FourierSM}) we show%
\begin{equation}
\left\vert \partial_{x}\mathcal T_{4}\right\vert \leq C\frac{1}{\sqrt{|t|}}.
\label{wi10}%
\end{equation}
Finally, we look to $\mathcal T_{5}.$ Again, using Parseval's identity and the
convolution theorem we calculate
\[
\mathcal T_{5}=\frac{\left(  2\pi\right)  ^{2}}{\sqrt{4\pi it}}\int_{-\infty}^{\infty
}e^{iz^{2}/4t}\int_{-\infty}^{\infty}K\left(  x,-z_{2}\right)  \left(
\int_{-\infty}^{\infty}F_{s}\left(  z-z_{2}-z_{1}\right)  K\left(
y,-z_{1}\right)  ^{\dagger}dz_{1}\right)  dz_{2}dz.
\]
Then, using (\ref{ESTK}), (\ref{ESTK1}) and (\ref{FourierSM}) we show that
\begin{equation}
\left\vert \partial_{x}^{j}\mathcal T_{5}\right\vert \leq C\frac{1}{\sqrt{|t|}},\text{
\ }j=0,1. \label{wi11}%
\end{equation}
By means of (\ref{wi0}) and the estimates (\ref{wi1}, \ref{wi2}, \ref{wi3},
\ref{wi6}, \ref{wi8}, \ref{wi9}, \ref{wi10}, \ref{wi11}) we deduce from
(\ref{eq4}, \ref{eq4.a}), and (\ref{T}) that
\[
\left\Vert e^{-itH}P_{\operatorname*{c}}\psi\right\Vert _{L^{\infty}}\leq
\frac{C}{\sqrt{t}}\left\Vert \psi\right\Vert _{L^{1}}%
\]
and, moreover,%
\[
\left\Vert e^{-itH}P_{\operatorname*{c}}\psi\right\Vert _{W_{1,\infty}}%
\leq\frac{C}{\sqrt{t}}\left\Vert \psi\right\Vert _{W_{1,1}},
\]
for all $t\neq0.$
\end{proof}

Let us now prove the $L^{2}-L^{2}$ estimate for $e^{-itH}P_{\operatorname*{c}%
}.$

\begin{proposition}
Suppose that the potential $V$ satisfies (\ref{PotentialHermitian}) and
(\ref{PotentialL11}). Then, there is $C>0,$ such that%
\begin{equation}
\left\Vert e^{-itH}P_{\operatorname*{c}}\right\Vert _{\mathcal{B}\left(
L^{2}\right)  }\leq1, \label{eq11bis}%
\end{equation}
and%
\begin{equation}
\left\Vert e^{-itH}P_{\operatorname*{c}}\right\Vert _{\mathcal{B}\left(
W_{1,2}^{A,B}\right)  }\leq C, \label{eq11}%
\end{equation}
holds for all $t\in\mathbb{R}.$
\end{proposition}

\begin{proof}
Estimate (\ref{eq11bis}) is consequence of the unitarity of $e^{-itH}$ in
$L^{2}.$ Further, by (\ref{norm}) and since $\sqrt{H+L} e^{-itH}
P_{\operatorname*{c}} =e^{-itH} P_{\operatorname*{c}} \sqrt{H+L},$ we have
that,
\begin{equation}
\left\Vert e^{-itH} P_{\operatorname*{c}} \right\Vert _{\mathcal{B}\left(
\mathcal{H }\right)  }\leq C. \label{eq10}%
\end{equation}
Then, using (\ref{equiv2}) we obtain (\ref{eq11}).
\end{proof}

\begin{proof}
[Proof of Theorem \ref{Theorem1}]The general estimate (\ref{estimate1}) is an
interpolation (use the Riesz-Thorin theorem, see \cite{ReedSimon}) between the
estimates (\ref{L1-Linf}) and (\ref{eq11bis}). 
Interpolating between estimates (\ref{L1-LinfSobolev}) and (\ref{eq11}), we attain
(\ref{estimate2}).
\end{proof}

\begin{proof}
[Proof of Theorem \ref{Theorem2}]The Strichartz estimates are deduced from the
$L^{p}-L^{p^{\prime}}$ estimates in Theorem \ref{Theorem1}. See the proof of
Theorem 2.3.3 on page 33 of \cite{Cazenave} for further details.
\end{proof}

\end{document}